\documentclass[submission,copyright,creativecommons]{eptcs}
\providecommand{\event}{ACT 2022}
\usepackage{underscore}
\usepackage[T1]{fontenc}

\usepackage{./sty/leosstylefile}
\usepackage{./sty/tikzit}
\usetikzlibrary{shapes}
\usetikzlibrary{arrows}

\tikzstyle{node}=[fill=black, draw=black, shape=circle, scale=0.5]
\tikzstyle{wnode}=[fill=white, draw=black, shape=circle, scale=0.5]
\tikzstyle{textbox}=[inner sep=2pt, shape=rectangle, fill=none]
\tikzstyle{textnode}=[inner sep=0mm, shape=circle, fill=white]
\tikzstyle{gnode}=[inner sep=0mm, minimum size=1mm, fill={rgb,255: red,221; green,221; blue,221}, draw={rgb,255: red,221; green,221; blue,221}, shape=circle]
\tikzstyle{refine}=[fill=black, draw=black, shape=regular polygon, regular polygon sides=3, rotate=180, scale=0.5]
\tikzstyle{coarsen}=[fill=white, draw=black, shape=regular polygon, regular polygon sides=3, scale=0.5]
\tikzstyle{bdytextbox}=[fill=white, draw=black, shape=rectangle]
\tikzstyle{redbox}=[fill=white, draw=red, shape=rectangle, text=red]
\tikzstyle{bluecirc}=[inner sep=1mm, fill=white, draw={rgb,255: red,4; green,51; blue,255}, shape=circle, text={rgb,255: red,4; green,51; blue,255}]
\tikzstyle{rednode}=[fill=red, draw=red, shape=circle, scale=0.5]
\tikzstyle{new style 0}=[fill=white, draw=red, shape=circle, scale=0.5]
\tikzstyle{bluenode}=[fill={rgb,255: red,4; green,51; blue,255}, draw={rgb,255: red,4; green,51; blue,255}, shape=circle, scale=0.5]
\tikzstyle{blacksq}=[fill=black, draw=black, shape=rectangle, scale=0.5]
\tikzstyle{bluetext}=[fill=none, draw=none, shape=rectangle, text={rgb,255: red,4; green,51; blue,255}]
\tikzstyle{reg}=[draw, fill=white, rounded rectangle, rounded rectangle left arc=none, minimum height=1em, minimum width=1em, node font={\scriptsize}]
\tikzstyle{coreg}=[draw, fill=white, rounded rectangle, rounded rectangle right arc=none, minimum height=1em, minimum width=1em, node font={\scriptsize}]

\tikzstyle{edge}=[-, draw=black]
\tikzstyle{diredge}=[->, draw=black]
\tikzstyle{dashed edge}=[-, dashed, dash pattern=on 1pt off 1.5pt, draw=black]
\tikzstyle{dirdash}=[->, dashed, dash pattern=on 2pt off 0.5pt, draw=black]
\tikzstyle{mapsto}=[{|->}, draw=black]
\tikzstyle{gray diredge}=[draw={rgb,255: red,221; green,221; blue,221}, ->]
\tikzstyle{dark grey dirdash}=[->, dashed, dash pattern=on 2pt off 0.5pt, draw={rgb,255: red,81; green,81; blue,81}]
\tikzstyle{doubedge}=[-, draw=black, double=none, double distance=3pt, inner sep=0pt, thick]
\tikzstyle{thedge}=[-, thick, draw=black]
\tikzstyle{gray dashed}=[-, dashed, dash pattern=on 1pt off 1.5pt, draw={rgb,255: red,128; green,128; blue,128}]
\tikzstyle{rededge}=[-, draw=red]
\tikzstyle{gray edge}=[-, draw={rgb,255: red,128; green,128; blue,128}]
\tikzstyle{blthedge}=[-, thick, draw={rgb,255: red,4; green,51; blue,255}]
\tikzstyle{blthdash}=[-, dashed, dash pattern=on 1pt off 1.5pt, thick, draw={rgb,255: red,4; green,51; blue,255}]
\tikzstyle{resistor}=[R]
\tikzstyle{inductor}=[L]
\tikzstyle{capacitor}=[C]
\tikzstyle{voltage-source}=[american voltage source]
\tikzstyle{current-source}=[american current source]
\tikzstyle{dirrededge}=[draw=red, ->]

\tikzstyle{object}=[inner sep=0mm, shape=circle, fill=none]
\tikzstyle{bullet}=[fill=black, draw=black, shape=circle, scale=0.3]
\tikzstyle{circ}=[fill=white, draw=black, shape=circle, scale=0.3]
\tikzstyle{objectbox}=[inner sep=1pt, shape=rectangle, fill=none]
\tikzstyle{bdyobjectbox}=[fill=white, draw=black, shape=rectangle]

\tikzstyle{morphism}=[->, draw=black]
\tikzstyle{dash morphism}=[->, dashed, dash pattern=on 2pt off 0.5pt, draw=black]
\tikzstyle{mapsto}=[{|->}, draw=black]
\tikzstyle{nat transf}=[-implies, double, double distance=3pt, thick]
\tikzstyle{gray nat transf}=[-implies, draw=gray, double, double distance=3pt, thick]
\tikzstyle{equality}=[-, double, double distance=3pt]
\tikzstyle{squig morphism}=[rightsquigarrow, draw=black]
\tikzstyle{hookarrow}=[right hook->, draw=black]

\definecolor{electroblue}{RGB}{4,51,255}
\usepackage{siunitx}
\usepackage[american,siunitx,]{circuitikz}
\ctikzset{
    color/.initial=electroblue,
    bipoles/thickness=3,
    bipoles/length=1.55cm,
    sources/scale=0.8,
    capacitors/scale=0.9,
}
\tikzstyle{vsource}=[rmeter, t={\textsf{\tiny -- +}}] 
\tikzstyle{ammeter}=[rmeter, t={\textsf{A}}] 
\tikzstyle{vmeter}=[rmeter, t={\textsf{V}}] 
\tikzstyle{elecdot}=[circle,fill,inner sep=0.85pt]

\newcommand{\includegraffle}[1]{
  {\lower10pt\hbox{$\includegraphics[height=1cm]{tikz/#1.pdf}$}}
}

\graphicspath{ {figures/} }

\title{String Diagrams for Layered Explanations}
\author{Leo Lobski \qquad\qquad Fabio Zanasi
\institute{University College London}
}
\def\titlerunning{String Diagrams for Layered Explanations}
\def\authorrunning{Lobski \& Zanasi}
\begin{document}
\maketitle

\begin{abstract}
We propose a categorical framework to reason about scientific \emph{explanations}: descriptions of a phenomenon meant to translate it into simpler terms, or into a context that has been already understood. Our motivating examples come from systems biology, electrical circuit theory, and concurrency. We demonstrate how three explanatory models in these seemingly diverse areas can be all understood uniformly via a graphical calculus of \emph{layered props}. Layered props allow for a compact visual presentation of the same phenomenon at different levels of precision, as well as the translation between these levels. Notably, our approach allows for \emph{partial explanations}, that is, for translating just one part of a diagram while keeping the rest of the diagram untouched. Furthermore, our approach paves the way for formal reasoning about counterfactual models in systems biology.
\end{abstract}

\mathversion{normal2}
\section{Introduction}\label{sec:introduction}

Different fields of science and engineering come with their own notions and traditions of explaining one phenomenon in terms of another one. For example, statistical mechanics explains thermodynamics, since it relies on fewer assumptions, which are moreover perceived as more fundamental than those of thermodynamics. A similar pattern may be found in the reduction of climate science to various areas of physics and biology. The converse move, from a ``lower'' to a ``higher'' level, is also interesting: for instance, temperature and vessel shape may be used to explain crystallisation. Choosing the right level of abstraction is paramount for successful communication between different disciplines, as well as between the scientific community and the general public. In particular, the definition of what constitutes an explanation is an increasingly important topic in the areas of automated reasoning and artificial intelligence~\cite{ExplainableAI}.

Perhaps the most drastic divide between different modes of explaining can be found in biology, where some phenomena are explained {\em mechanistically} (or {\em reductively}), that is, by reducing them to the underlying chemical or physical laws, while others are explained {\em functionally}, that is, by appealing to what an organism does as a part of a larger whole~\cite{krivine-siglog,rosen-life-itself}. For instance, when explaining production of ATP within a cell, the mitochondria can either be introduced as elementary blocks providing energy to the cell (functional), or as compartments containing a whole pathway to process ATP (mechanistic). This divide is not merely of conceptual interest, but has practical implications for the modelling of biological systems: the ability to replicate biological functions is taken as a measure of success of the rule based models \cite{signal-transduction,krivine-siglog}. However, the existing rule based languages that model molecular interactions are typically not able to formally distinguish between mechanistic and functional rules, as these exist at different levels of abstraction \cite{krivine-siglog}.

The goal of this work is to identify fundamental mathematical structures underlying explanations across different fields of science. Upon these structures, we develop a formalism that is able to describe the different levels of abstractions involved in an explanation, and account for more elaborate aspects such as the divide above. Additionally, we attempt to provide a uniform framework for {\em counterfactual reasoning} by allowing explanations that depend on what could potentially occur. Ability to model counterfactual dependencies is of interest in rule-based models of molecular interactions~\cite{counterfactual}. We shall illustrate our approach by showing how it models case studies in diverse scientific areas.\footnote{On the other hand, we do not delve into the philosophical ramifications of our approach. Rather, the aim is to offer an abstract formalisation of of existing intuitions, thus potentially providing precise tools for debating what a scientific explanation {\em should} be.}

In our framework, explanations always concern a certain {\em process}. The process can be thought of as an actually occurring natural phenomenon, or a computation, or a rule in some formal system. An explanation should then consist of another process whose level of abstraction is strictly lower than that of the process being explained. In addition to the lower level process, an explanation should state in what way the two processes are related, for example by giving a translation from one to the other. Moreover, we want the explanations to be {\em modular} or {\em compositional}, in the sense that the same explanation may be reused multiple times in case different systems have equivalent subsystems, and that the explanations can be composed to create larger, more complicated explanations. The reason for requiring modularity is twofold. First, it allows explanations to be reused by potentially different areas, in much the same way lemmas and theorems in mathematics are used to develop different theories. Second, this allows for a certain efficiency, as we may be interested in explaining only a part of a large system; in such a case modularity allows us to focus on this one part only, instead of explaining the whole system.

The above requirements for what an explanation should be like lead naturally to {\em monoidal categories}, as these allow for both sequential and parallel composition of processes (i.e.~morphisms in a category). We assume that the monoidal categories are partially ordered ``by abstraction'', so that more abstract theories (i.e.~categories) are higher in the order. We want to be able to compose not just the processes but also the explanations, so that we require the categories and functors under consideration to have a monoidal structure. We thus arrive at a $2$-category which is able to simultaneously talk about processes in all the individual categories ($0$-cells), translations between processes ($1$-cells), compositions of the processes and the translations, as well as rules or equations between the processes and the translations ($2$-cells). The definitions of an explanation (Definitions~\ref{def:explanation-1} and~\ref{def:explanation-2}) use all of this structure. This is the motivation for what we call a {\em layered prop} (Definition~\ref{def:layered-prop}).

It is worth noting that, in the categorical approaches inspired by the paradigm of functorial semantics, an explanation and what is being explained live in two separate categories, with some translation between them expressed as a functor --- see e.g.~\cite{compositional-networks,graphical-affine-algebra,electrical-circuits}. Within this perspective, some equality or relation in the domain is explained by passing to the codomain (or vice versa). Our framework allows to treat such situations in a single language, staying within one category. The main technical advantage of our approach is that partial interpretations are built into our language from the very beginning, potentially reducing the amount of computation that is needed. More conceptually, unlike in the functorial semantics approach, working in our framework allows for counterfactual reasoning: since we can mix-and-match categories and morphisms, this gives the flexibility to ask such questions as {\em What would happen if $p$ did not occur?}

Our contributions are organised as follows. In Section~\ref{sec:syntax} we define layered props and outline their connection to the so-called internal string diagram construction. Section~\ref{sec:semantics} briefly outlines how a layered prop can be interpreted in the bicategory of pointed profunctors. We give three definitions of an explanation in Section~\ref{sec:explanations}: one applies to $1$-cells, another one to $2$-cells, and the last one formalises counterfactual explanations. The remaining sections contain case studies formalised in our framework. Section~\ref{sec:glucose} shows an example involving biology and chemistry. Section~\ref{sec:electrical-circuits} shows the explanation of electrical circuit behaviour in terms of signal flow graphs --- as it draws from the circuit theory developed in~\cite{electrical-circuits,graphical-affine-algebra}, this example also clarifies how our approach relates to the `functorial semantics' approaches. Finally, Section~\ref{sec:concurrency} presents a case study from concurrency, involving the explanation of CCS expressions.

\section{Layered Props}\label{sec:syntax}

We shall build our language on \emph{string diagrammatic} syntax: the standard representation of morphisms in (strict) monoidal categories~\cite{selinger}. Algebraic reasoning on string diagrams is typically formulated using props (\textbf{pro}duct and \textbf{p}ermutation categories), which are just symmetric strict monoidal categories with the natural numbers as objects --- see e.g.~\cite{maclane-cat-algebra,lack-props,zanasi-thesis} for an overview. In fact, in order to model the different layers involved in an explanation, we will need a more sophisticated concept: {\em layered props}.

Since we want to talk about ``string diagrams in context'', the context being a theory at a particular level of abstraction, we draw the string diagrams inside a rectangle which represents its context. This allows us to reason both {\em internally} with the string diagrams, as well as {\em externally} by pasting and piling the rectangles. In order to formalise such graphical intuition as a layered prop, we need the preliminary notions of a {\em system of sets} and a {\em layered monoidal theory}.

We begin with systems of sets, which we think of as contexts and translations between them. Fix a collection of sets $\Omega$. An {\em $\Omega$-type} is a finite list of pairs $(\omega_1,\alpha_1;\dots;\omega_n,\alpha_n)$ where each $\omega_i$ is in $\Omega$ and each $\alpha_i\in\omega_i^*$ is an element in the free monoid on $\omega_i$. Precisely, define $\Omega$-types recursively~as:
\begin{itemize}[topsep=0pt,itemsep=-1ex,partopsep=1ex,parsep=1ex]
  \item the empty list $\varepsilon$ is an $\Omega$-type,
  \item if $t$ is a type, $\omega\in\Omega$, and $\alpha\in\omega^*$, then $(t; \ \omega,\alpha)$ is an $\Omega$-type.
\end{itemize}
We denote the collection of all $\Omega$-types by $\mathtt{type}^\Omega$.

We call $\Omega$ a {\em system of sets} when it is equipped with a partial order, and, for each comparable pair $\omega\leq\omega'$, with a choice of a homomorphism $f:\omega'^*\rightarrow\omega^*$. Intuitively, as we think of the sets $\omega\in\Omega$ as contexts, the partial order is saying which contexts are more abstract, and the homomorphisms are translations from more abstract contexts to less abstract ones. We now introduce the counterpart of algebraic theories for monoidal categories (typically called monoidal theories, see e.g.~\cite{zanasi-thesis}) based on this structure.
\begin{definition}[Layered monoidal theory]
  A {\em layered monoidal theory} is a tuple $(\Omega,\Sigma,\arity,\coarity)$ consisting of a system of sets $\Omega$, a set $\Sigma$ ({\em signature}), and functions $\arity,\coarity:\Sigma\rightarrow\mathtt{type}^\Omega$.
\end{definition}
It is convenient to introduce notation for the {\em internal signature} $\Sigma^i$, defined as
$$\Sigma^i\coloneqq\left\{\sigma\in\Sigma : \text{ there are }\omega\in\Omega\text{ and }\alpha,\beta\in\omega^* \text{ s.t. } \arity(\sigma)=(\omega,\alpha) \text{ and } \coarity(\sigma)=(\omega,\beta)\right\}.$$
The idea is that the generators in $\Sigma^i$ are completely contained in a single context $\omega$: there is no transition between contexts involved. We define the {\em terms} and the corresponding {\em sorts} (arity-coarity pairs of types $(t\mid s)$) of a layered monoidal theory by the recursive procedure in Figure~\ref{fig:terms}. For the $\otimes_\omega$-rule, there is a side condition that only the rules for $\Sigma^i$, identity, composition and $\otimes_{\omega}$ are used in constructing the terms $x$ and $y$. This ensures that $x$ and $y$ only contain generators from the internal signature, so that it makes sense to graphically represent the term $x\otimes_{\omega}y$ as juxtaposition of $x$ and $y$ inside the rectangle representing $\omega$. We call the terms that are generated using only these four rules {\em internal}. If a layered monoidal theory is generated by monoidal categories (see Section~\ref{subsec:from-monoidal} below), the internal terms will correspond precisely to morphisms inside the categories.
\bgroup
\def\arraystretch{1.75} 
\begin{figure}[h]
\scriptsize
\centering
\begin{tabular}{ c c c }
  \\
\begin{tabular}{ c }
$\sigma\in\Sigma\setminus\Sigma^i$ \\
\hline
$\sigma : (\arity(\sigma)\mid\coarity(\sigma))$
\end{tabular}
&
\begin{tabular}{ c }
$\sigma\in\Sigma^i\quad \arity(\sigma)=\omega,\alpha\quad \coarity(\sigma)=\omega,\beta$ \\
\hline
$\scalebox{.6}{\begin{tikzpicture}
	\begin{pgfonlayer}{nodelayer}
		\node [style=none] (0) at (-2, 1) {};
		\node [style=none] (1) at (-2, -1) {};
		\node [style=none] (2) at (2, 1) {};
		\node [style=none] (3) at (2, -1) {};
		\node [style=none] (4) at (-2, 0) {};
		\node [style=none] (5) at (2, 0) {};
		\node [style=redbox] (6) at (0, 0) {$\sigma$};
		\node [style=none] (7) at (-2.75, 0) {$\alpha$};
		\node [style=none] (8) at (-3.5, 0) {$\omega$};
		\node [style=none] (9) at (2.75, 0) {$\beta$};
		\node [style=none] (10) at (3.5, 0) {$\omega$};
	\end{pgfonlayer}
	\begin{pgfonlayer}{edgelayer}
		\draw [style=rededge] (4.center) to (6);
		\draw [style=rededge] (6) to (5.center);
		\draw [style=edge] (0.center) to (1.center);
		\draw [style=edge] (1.center) to (3.center);
		\draw [style=edge] (3.center) to (2.center);
		\draw [style=edge] (2.center) to (0.center);
	\end{pgfonlayer}
\end{tikzpicture}} : (\omega,\alpha\mid\omega,\beta)$
\end{tabular}
&
\begin{tabular}{ c }
\hline
$\scalebox{.6}{\begin{tikzpicture}
	\begin{pgfonlayer}{nodelayer}
		\node [style=none] (6) at (-2, 0) {};
		\node [style=none] (7) at (2, 0) {};
		\node [style=none] (0) at (-2, 1) {};
		\node [style=none] (1) at (-2, -1) {};
		\node [style=none] (2) at (2, -1) {};
		\node [style=none] (3) at (2, 1) {};
		\node [style=none] (4) at (3.5, 0) {$\omega$};
		\node [style=none] (5) at (-3.5, 0) {$\omega$};
		\node [style=none] (8) at (-2.75, 0) {$\alpha$};
		\node [style=none] (9) at (2.75, 0) {$\alpha$};
	\end{pgfonlayer}
	\begin{pgfonlayer}{edgelayer}
		\draw [style=rededge] (6.center) to (7.center);
		\draw [style=edge] (0.center) to (1.center);
		\draw [style=edge] (1.center) to (2.center);
		\draw [style=edge] (2.center) to (3.center);
		\draw [style=edge] (3.center) to (0.center);
	\end{pgfonlayer}
\end{tikzpicture}} : (\omega,\alpha \mid \omega,\alpha)$
\end{tabular}
\end{tabular}
\begin{tabular}{ c c c c }
  \\
  \begin{tabular}{ c }
    \hline
    $\scalebox{.6}{\begin{tikzpicture}
	\begin{pgfonlayer}{nodelayer}
		\node [style=none] (6) at (-2, 1) {};
		\node [style=none] (7) at (2, -1) {};
		\node [style=none] (0) at (-2, 1.75) {};
		\node [style=none] (1) at (-2, 0.25) {};
		\node [style=none] (2) at (2, -1.75) {};
		\node [style=none] (3) at (2, -0.25) {};
		\node [style=none] (4) at (3.5, -1) {$\omega$};
		\node [style=none] (5) at (-3.5, 1) {$\omega$};
		\node [style=none] (8) at (-2.75, 1) {$\alpha$};
		\node [style=none] (9) at (2, 1) {};
		\node [style=none] (10) at (-2, -1) {};
		\node [style=none] (11) at (2, 1.75) {};
		\node [style=none] (12) at (2, 0.25) {};
		\node [style=none] (13) at (-2, -1.75) {};
		\node [style=none] (14) at (-2, -0.25) {};
		\node [style=none] (15) at (-3.5, -1) {$\tau$};
		\node [style=none] (16) at (3.5, 1) {$\tau$};
		\node [style=none] (17) at (-2.75, -1) {$\beta$};
		\node [style=none] (18) at (2.75, 1) {$\beta$};
		\node [style=none] (19) at (2.75, -1) {$\alpha$};
		\node [style=none] (20) at (0, 2.25) {};
	\end{pgfonlayer}
	\begin{pgfonlayer}{edgelayer}
		\draw [style=rededge] (6.center) to (7.center);
		\draw [style=rededge] (9.center) to (10.center);
		\draw [style=edge] (11.center) to (12.center);
		\draw [style=edge] (12.center) to (13.center);
		\draw [style=edge] (13.center) to (14.center);
		\draw [style=edge] (14.center) to (11.center);
		\draw [style=gray edge] (0.center) to (1.center);
		\draw [style=gray edge] (0.center) to (3.center);
		\draw [style=gray edge] (1.center) to (2.center);
		\draw [style=gray edge] (3.center) to (2.center);
	\end{pgfonlayer}
\end{tikzpicture}} : (\omega,\alpha ; \tau,\beta \mid \tau,\beta ; \omega,\alpha)$
    \end{tabular}
    &
\begin{tabular}{ c }
\hline
$\scalebox{.6}{\begin{tikzpicture}
	\begin{pgfonlayer}{nodelayer}
		\node [style=none] (0) at (-1, -1) {};
		\node [style=none] (1) at (-1, 1) {};
		\node [style=none] (2) at (1, 1) {};
		\node [style=none] (3) at (1, -1) {};
	\end{pgfonlayer}
	\begin{pgfonlayer}{edgelayer}
		\draw [style=dashed edge] (1.center) to (0.center);
		\draw [style=dashed edge] (0.center) to (3.center);
		\draw [style=dashed edge] (3.center) to (2.center);
		\draw [style=dashed edge] (2.center) to (1.center);
	\end{pgfonlayer}
\end{tikzpicture}} : (\varepsilon\mid\varepsilon)$
\end{tabular}
&
\begin{tabular}{ c }
\hline
$\scalebox{.6}{\begin{tikzpicture}
	\begin{pgfonlayer}{nodelayer}
		\node [style=none] (0) at (-0.25, 1) {};
		\node [style=none] (1) at (-0.25, -1) {};
		\node [style=none] (2) at (0.75, -1) {};
		\node [style=none] (3) at (0.75, 1) {};
		\node [style=none] (4) at (1.5, 0) {$\varepsilon$};
		\node [style=none] (5) at (2.5, 0) {$\omega$};
		\node [style=none] (6) at (-1.75, 0) {$\varepsilon$};
	\end{pgfonlayer}
	\begin{pgfonlayer}{edgelayer}
		\draw [style=edge, bend right=75, looseness=1.25] (0.center) to (1.center);
		\draw [style=edge] (1.center) to (2.center);
		\draw [style=edge] (2.center) to (3.center);
		\draw [style=edge] (3.center) to (0.center);
	\end{pgfonlayer}
\end{tikzpicture}} : (\varepsilon\mid\omega,\varepsilon)$
\end{tabular}
&
\begin{tabular}{ c }
\hline
$\scalebox{.6}{\begin{tikzpicture}
	\begin{pgfonlayer}{nodelayer}
		\node [style=none] (0) at (0.5, 1) {};
		\node [style=none] (1) at (0.5, -1) {};
		\node [style=none] (2) at (-0.5, -1) {};
		\node [style=none] (3) at (-0.5, 1) {};
		\node [style=none] (4) at (-1.25, 0) {$\varepsilon$};
		\node [style=none] (5) at (-2.25, 0) {$\omega$};
		\node [style=none] (6) at (2, 0) {$\varepsilon$};
	\end{pgfonlayer}
	\begin{pgfonlayer}{edgelayer}
		\draw [style=edge, bend left=75, looseness=1.25] (0.center) to (1.center);
		\draw [style=edge] (1.center) to (2.center);
		\draw [style=edge] (2.center) to (3.center);
		\draw [style=edge] (3.center) to (0.center);
	\end{pgfonlayer}
\end{tikzpicture}} : (\omega,\varepsilon\mid\varepsilon)$
\end{tabular}
\end{tabular}
\begin{tabular}{ c c }
  \\
\begin{tabular}{ c }
\hline
$\scalebox{.6}{\begin{tikzpicture}
	\begin{pgfonlayer}{nodelayer}
		\node [style=none] (0) at (-1.5, 0.75) {};
		\node [style=none] (1) at (-1.5, 2.75) {};
		\node [style=none] (2) at (-1.5, -0.75) {};
		\node [style=none] (3) at (-1.5, -2.75) {};
		\node [style=none] (4) at (2.5, 1) {};
		\node [style=none] (5) at (2.5, -1) {};
		\node [style=none] (6) at (0.5, 1.5) {};
		\node [style=none] (7) at (0.5, -1.5) {};
		\node [style=none] (8) at (-1.5, 1.75) {};
		\node [style=none] (10) at (2.5, 0.5) {};
		\node [style=none] (12) at (0.25, 1) {};
		\node [style=none] (13) at (-1.5, -1.75) {};
		\node [style=none] (14) at (2.5, -0.5) {};
		\node [style=none] (15) at (0.25, -1) {};
		\node [style=none] (16) at (-3, 1.75) {$\omega$};
		\node [style=none] (17) at (-3, -1.75) {$\omega$};
		\node [style=none] (18) at (4, 0) {$\omega$};
		\node [style=none] (19) at (-2, 1.75) {};
		\node [style=none] (20) at (-2, 1.75) {$\alpha$};
		\node [style=none] (21) at (3, 0.5) {$\alpha$};
		\node [style=none] (22) at (3, -0.5) {$\beta$};
		\node [style=none] (23) at (-2, -1.75) {$\beta$};
		\node [style=none] (24) at (0, 3.25) {};
	\end{pgfonlayer}
	\begin{pgfonlayer}{edgelayer}
		\draw [style=edge] (5.center) to (4.center);
		\draw [style=edge] (1.center) to (0.center);
		\draw [style=edge] (3.center) to (2.center);
		\draw [style=edge, bend left=90, looseness=2.00] (0.center) to (2.center);
		\draw [style=edge, bend left=15] (1.center) to (6.center);
		\draw [style=edge, bend right=15] (6.center) to (4.center);
		\draw [style=edge, bend right=15] (3.center) to (7.center);
		\draw [style=edge, bend left=15] (7.center) to (5.center);
		\draw [style=rededge, bend left=15] (8.center) to (12.center);
		\draw [style=rededge, bend right=15, looseness=0.75] (12.center) to (10.center);
		\draw [style=rededge, bend right=15] (13.center) to (15.center);
		\draw [style=rededge, bend left=15, looseness=0.75] (15.center) to (14.center);
	\end{pgfonlayer}
\end{tikzpicture}} : (\omega,\alpha;\omega,\beta \mid \omega,\alpha\beta)$
\end{tabular}
&
\begin{tabular}{ c }
\hline
$\scalebox{.6}{\begin{tikzpicture}
	\begin{pgfonlayer}{nodelayer}
		\node [style=none] (0) at (2.5, 0.75) {};
		\node [style=none] (1) at (2.5, 2.75) {};
		\node [style=none] (2) at (2.5, -0.75) {};
		\node [style=none] (3) at (2.5, -2.75) {};
		\node [style=none] (4) at (-1.5, 1) {};
		\node [style=none] (5) at (-1.5, -1) {};
		\node [style=none] (6) at (0.5, 1.5) {};
		\node [style=none] (7) at (0.5, -1.5) {};
		\node [style=none] (8) at (2.5, 1.75) {};
		\node [style=none] (10) at (-1.5, 0.5) {};
		\node [style=none] (12) at (0.75, 1) {};
		\node [style=none] (13) at (2.5, -1.75) {};
		\node [style=none] (14) at (-1.5, -0.5) {};
		\node [style=none] (15) at (0.75, -1) {};
		\node [style=none] (16) at (4, 1.75) {$\omega$};
		\node [style=none] (17) at (4, -1.75) {$\omega$};
		\node [style=none] (18) at (-3, 0) {$\omega$};
		\node [style=none] (19) at (3, 1.75) {};
		\node [style=none] (20) at (3, 1.75) {$\alpha$};
		\node [style=none] (21) at (-2, 0.5) {$\alpha$};
		\node [style=none] (22) at (-2, -0.5) {$\beta$};
		\node [style=none] (23) at (3, -1.75) {$\beta$};
		\node [style=none] (24) at (0, 3.25) {};
	\end{pgfonlayer}
	\begin{pgfonlayer}{edgelayer}
		\draw [style=edge] (5.center) to (4.center);
		\draw [style=edge] (1.center) to (0.center);
		\draw [style=edge] (3.center) to (2.center);
		\draw [style=edge, bend right=90, looseness=2.00] (0.center) to (2.center);
		\draw [style=edge, bend right=15] (1.center) to (6.center);
		\draw [style=edge, bend left=15] (6.center) to (4.center);
		\draw [style=edge, bend left=15] (3.center) to (7.center);
		\draw [style=edge, bend right=15] (7.center) to (5.center);
		\draw [style=rededge, bend right=15] (8.center) to (12.center);
		\draw [style=rededge, bend left=15, looseness=0.75] (12.center) to (10.center);
		\draw [style=rededge, bend left=15] (13.center) to (15.center);
		\draw [style=rededge, bend right=15, looseness=0.75] (15.center) to (14.center);
	\end{pgfonlayer}
\end{tikzpicture}} : (\omega,\alpha\beta \mid \omega,\alpha;\omega,\beta)$
\end{tabular}
\end{tabular}
\begin{tabular}{ c c }
\begin{tabular}{ c }
$f:\omega\rightarrow\tau$ \\
\hline
$\scalebox{.6}{\begin{tikzpicture}
	\begin{pgfonlayer}{nodelayer}
		\node [style=none] (0) at (-3, 0) {$\omega$};
		\node [style=none] (1) at (-2, 1) {};
		\node [style=none] (2) at (-2, -1) {};
		\node [style=none] (3) at (0, 1) {};
		\node [style=none] (4) at (0, -1) {};
		\node [style=none] (5) at (2, 1) {};
		\node [style=none] (6) at (2, -1) {};
		\node [style=none] (7) at (3, 0) {$\tau$};
		\node [style=none] (8) at (0, 1.5) {$\refine$};
		\node [style=none] (9) at (-2, 0) {};
		\node [style=none] (10) at (2, 0) {};
		\node [style=none] (11) at (-1, 0.25) {$\alpha$};
		\node [style=none] (12) at (1, 0.25) {$f\alpha$};
		\node [style=none] (13) at (0, 2) {};
	\end{pgfonlayer}
	\begin{pgfonlayer}{edgelayer}
		\draw [style=rededge] (9.center) to (10.center);
		\draw [style=edge] (3.center) to (4.center);
		\draw [style=edge] (1.center) to (5.center);
		\draw [style=edge] (1.center) to (2.center);
		\draw [style=edge] (2.center) to (6.center);
		\draw [style=edge] (6.center) to (5.center);
	\end{pgfonlayer}
\end{tikzpicture}} : (\omega,\alpha \mid \tau,f\alpha)$
\end{tabular}
&
\begin{tabular}{ c }
$f:\omega\rightarrow\tau$ \\
\hline
$\scalebox{.6}{\begin{tikzpicture}
	\begin{pgfonlayer}{nodelayer}
		\node [style=none] (0) at (3, 0) {$\omega$};
		\node [style=none] (1) at (2, 1) {};
		\node [style=none] (2) at (2, -1) {};
		\node [style=none] (3) at (0, 1) {};
		\node [style=none] (4) at (0, -1) {};
		\node [style=none] (5) at (-2, 1) {};
		\node [style=none] (6) at (-2, -1) {};
		\node [style=none] (7) at (-3, 0) {$\tau$};
		\node [style=none] (8) at (0, 1.5) {$\coarsen$};
		\node [style=none] (9) at (2, 0) {};
		\node [style=none] (10) at (-2, 0) {};
		\node [style=none] (11) at (1, 0.25) {$\alpha$};
		\node [style=none] (12) at (-1, 0.25) {$f\alpha$};
		\node [style=none] (13) at (0, 2) {};
	\end{pgfonlayer}
	\begin{pgfonlayer}{edgelayer}
		\draw [style=rededge] (9.center) to (10.center);
		\draw [style=edge] (3.center) to (4.center);
		\draw [style=edge] (5.center) to (1.center);
		\draw [style=edge] (1.center) to (2.center);
		\draw [style=edge] (2.center) to (6.center);
		\draw [style=edge] (6.center) to (5.center);
	\end{pgfonlayer}
\end{tikzpicture}} : (\tau,f\alpha \mid \omega,\alpha)$
\end{tabular}
\end{tabular}
\begin{tabular}{ c c c }
  \\
  \begin{tabular}{ c }
  $x : (t\mid s)$\quad $y : (s\mid u)$ \\
  \hline
  $x;y : (t,u)$
  \end{tabular}
&
\begin{tabular}{ c }
$x : (\omega,\alpha \mid \omega,\gamma)$\quad $y : (\omega,\beta \mid \omega,\delta)$ \\
\hline
$x\otimes_{\omega} y : (\omega,\alpha\beta \mid \omega,\gamma\delta)$
\end{tabular}
&
\begin{tabular}{ c }
$x : (t \mid s)$\quad $y : (u \mid w)$ \\
\hline
$x\otimes y : (t;u \mid s;w)$
\end{tabular}
\end{tabular}
\caption{Recursive construction of the terms of a layered monoidal theory. Each term of the sort $(\omega,\alpha \mid \tau,\beta)$ is drawn as an area connecting the type $\omega,\alpha$ on the left to the type $\tau,\beta$ on the right. The area inside a term, demarcated by black lines, is to be thought as representing the set $\omega$, and an internal red wire as $\alpha$ (the element of $\omega^*$). The change of type $\alpha\rightarrow\beta$ inside $\omega$ is drawn as a red box. The change of type at the level of sets $\omega\rightarrow\tau$ is drawn as a vertical black line. \label{fig:terms}}
\end{figure}
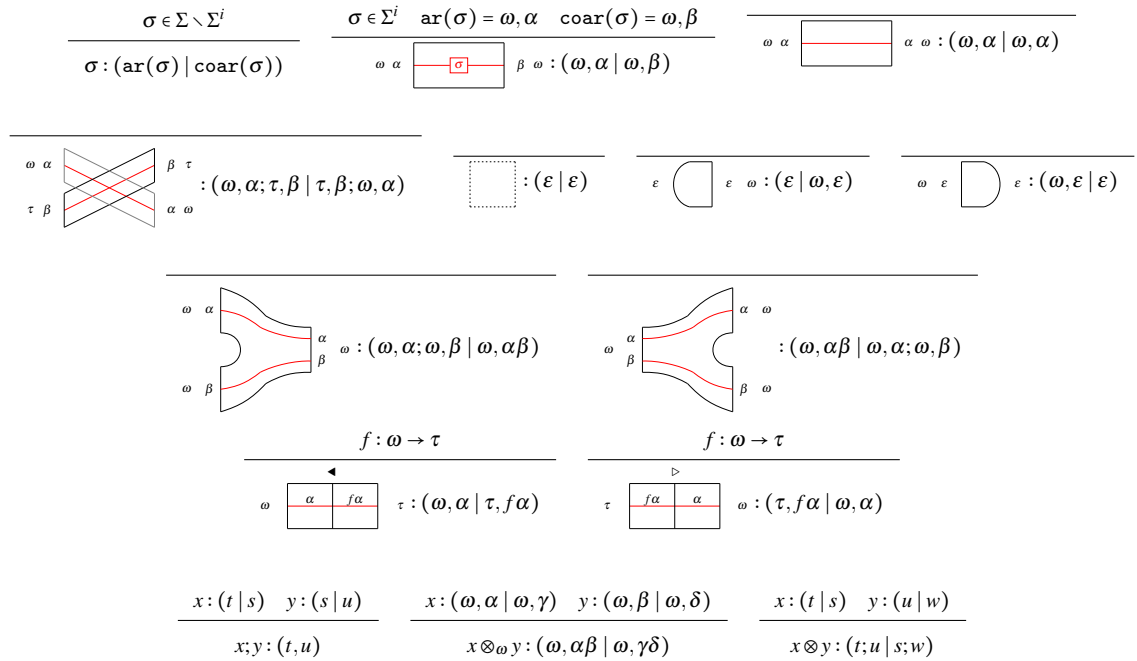
\egroup

We think of the {\em pants} and the {\em copants} (line 3 of Figure~\ref{fig:terms}) as composition and decomposition within a level of abstraction. The black and white triangles (line 4 of Figure~\ref{fig:terms}) are translations between the levels: $\refine$ translates an abstract layer to a more concrete one ({\em refinement}), while $\coarsen$ maps towards a higher abstraction ({\em coarsening}). In the pointed profunctor semantics (Section~\ref{sec:semantics}), pants will be interpreted as the monoidal product (seen as a profunctor), and copants as its adjoint profunctor (cf.~axioms in Figure~\ref{fig:axioms-pants2}). Likewise, refinement will be interpreted as a monoidal functor (seen as a profunctor), and coarsening as its adjoint (cf.~axioms in Figure~\ref{fig:axioms-pants3}).

In order to define a layered prop, we need to consider the terms modulo certain equations. Given $\omega\in\Omega$ and $\alpha,\beta\in\omega^*$, consider the internal terms with the sort $(\omega,\alpha\mid\omega,\beta)$. We may quotient this subset by the usual rules of monoidal categories: the identities and the monoidal unit are given by the third rule on the first line,
\begin{center}
\scalebox{.6}{\begin{tikzpicture}
	\begin{pgfonlayer}{nodelayer}
		\node [style=none] (6) at (-12, 0) {};
		\node [style=none] (7) at (-8, 0) {};
		\node [style=none] (8) at (-12, 1) {};
		\node [style=none] (9) at (-12, -1) {};
		\node [style=none] (10) at (-8, -1) {};
		\node [style=none] (11) at (-8, 1) {};
		\node [style=none] (12) at (-6.5, 0) {$\omega$};
		\node [style=none] (13) at (-13.5, 0) {$\omega$};
		\node [style=none] (14) at (-12.75, 0) {$\alpha$};
		\node [style=none] (15) at (-7.25, 0) {$\alpha$};
		\node [style=none] (16) at (-2, 0) {};
		\node [style=none] (17) at (2, 0) {};
		\node [style=none] (18) at (-2, 1) {};
		\node [style=none] (19) at (-2, -1) {};
		\node [style=none] (20) at (2, -1) {};
		\node [style=none] (21) at (2, 1) {};
		\node [style=none] (22) at (3.5, 0) {$\omega$};
		\node [style=none] (23) at (-3.5, 0) {$\omega$};
		\node [style=none] (24) at (-2.75, 0) {$\beta$};
		\node [style=none] (25) at (2.75, 0) {$\beta$};
		\node [style=none] (28) at (8, 1) {};
		\node [style=none] (29) at (8, -1) {};
		\node [style=none] (30) at (12, -1) {};
		\node [style=none] (31) at (12, 1) {};
		\node [style=none] (32) at (13.5, 0) {$\omega$};
		\node [style=none] (33) at (6.5, 0) {$\omega$};
		\node [style=none] (34) at (7.25, 0) {$\varepsilon$};
		\node [style=none] (35) at (12.75, 0) {$\varepsilon$};
	\end{pgfonlayer}
	\begin{pgfonlayer}{edgelayer}
		\draw [style=rededge] (6.center) to (7.center);
		\draw [style=edge] (8.center) to (9.center);
		\draw [style=edge] (9.center) to (10.center);
		\draw [style=edge] (10.center) to (11.center);
		\draw [style=edge] (11.center) to (8.center);
		\draw [style=rededge] (16.center) to (17.center);
		\draw [style=edge] (18.center) to (19.center);
		\draw [style=edge] (19.center) to (20.center);
		\draw [style=edge] (20.center) to (21.center);
		\draw [style=edge] (21.center) to (18.center);
		\draw [style=edge] (28.center) to (29.center);
		\draw [style=edge] (29.center) to (30.center);
		\draw [style=edge] (30.center) to (31.center);
		\draw [style=edge] (31.center) to (28.center);
	\end{pgfonlayer}
\end{tikzpicture}}
\end{center}
while the monoidal product $\otimes_{\omega}$ is represented by vertical juxtaposition inside the $\omega$-rectangle. Further, we may quotient all the terms with the sort $(t\mid s)$ by the usual rules of symmetric monoidal categories: the identities are given by appropriate vertical juxtapositions of terms generated by the third rule on the first line, the monoidal unit is given by the second rule on the second line, and the monoidal product $\otimes$ is once again represented by vertical juxtaposition, this time of whole rectangles.

\begin{definition}[Layered prop]\label{def:layered-prop}
  A {\em layered prop} generated by a layered monoidal theory $(\Omega,\Sigma,\arity,\coarity)$ is a $2$-category whose $0$-cells are the types $\mathtt{type}^\Omega$ and whose $1$-cells $t\rightarrow s$ are terms with sort $(t\mid s)$ quotiented by the laws of symmetric monoidal categories both internally and externally, as discussed above. The $2$-cells are generated by the rules in Figures~\ref{fig:axioms-pants1}, \ref{fig:axioms-pants2} and~\ref{fig:axioms-pants3}. Where arrows are going in both directions, we require the 2-cells to be inverses. Further, we require the usual triangle identities to hold for each unit-counit pair in Figure~\ref{fig:axioms-pants2}, and the usual laws of monoidal categories to hold for the isomorphisms in Figure~\ref{fig:axioms-pants3}.
\end{definition}

\begin{figure}[h]
    \centering
    \resizebox{\textwidth}{!}{
        \input{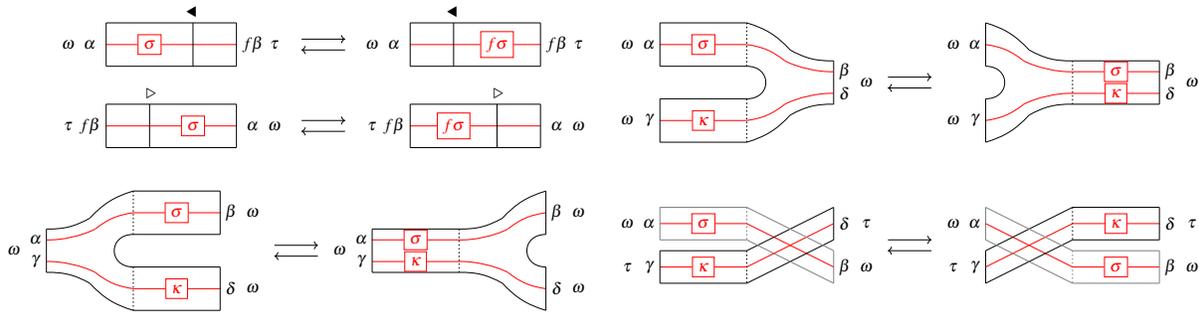}
    }
    \caption{2-cells of a layered prop expressing functoriality of refinement, coarsening, pants and copants.\label{fig:axioms-pants1}}
\end{figure}
\begin{figure}[h]
    \centering
    \resizebox{\textwidth}{!}{
        \input{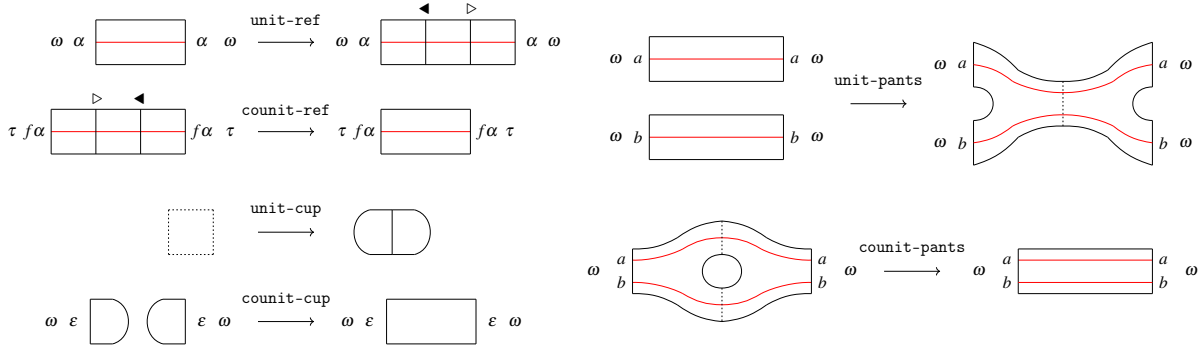}
    }
    \caption{2-cells of a layered prop that exhibit pants-copants and refinement-coarsening as two adjoint pairs.\label{fig:axioms-pants2}}
\end{figure}
\begin{figure}[h]
    \centering
    \resizebox{\textwidth}{!}{
        \input{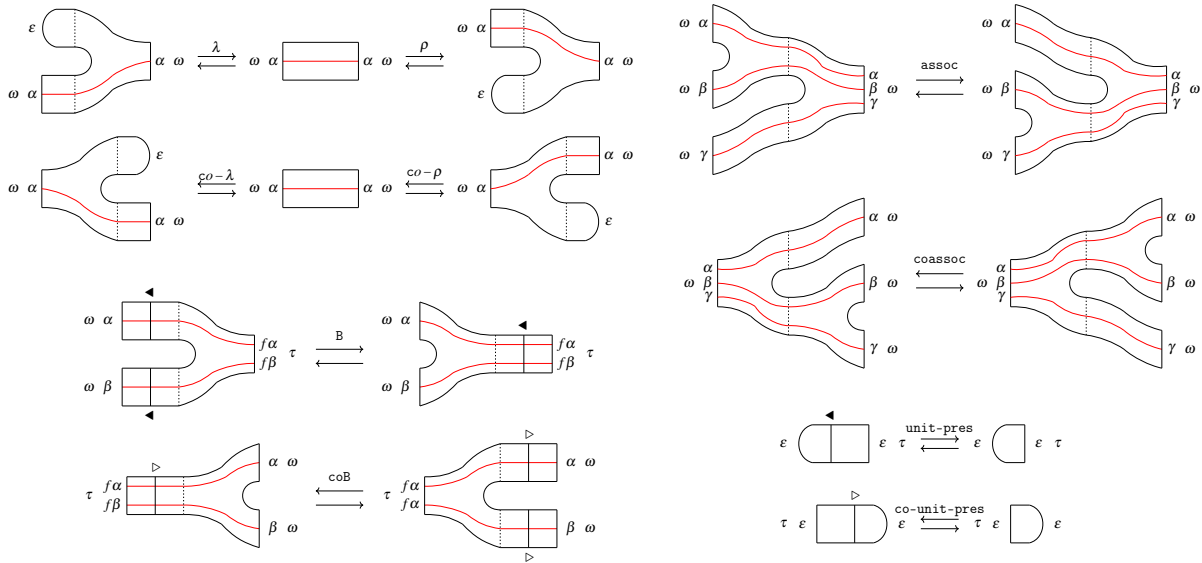}
    }
    \caption{2-cells of a layered prop that are motivated by monoidal categories and functors.\label{fig:axioms-pants3}}
\end{figure}
Note that the $1$-categorical structure of a layered prop can be seen as a generalisation of a coloured prop: any coloured prop gives rise to a layered prop with just one layer (i.e.~with just one set in $\Omega$). Furthermore, layered props are known in the literature as the {\em internal string diagram construction}. This was first introduced in the work of Bartlett, Douglas, Schommer-Pries and Vicary on topological quantum field theories~\cite{modular-categories}. The connection to profunctors is discussed by Hu~\cite{hu-thesis}.

\subsection{Layered Props from Monoidal Categories}\label{subsec:from-monoidal}
It is natural to build layered props from existing monoidal categories. In fact all the examples of layered props we consider arise in this way --- see Sections~\ref{sec:glucose}-\ref{sec:concurrency} below. We assume that instead of a system of sets, we have a system of {\em monoidal categories} $\Omega$ with monoidal functors instead of homomorphisms. The construction of the layered prop $\mathcal L(\Omega)$ then proceeds as before, taking the internal signature to contain all morphisms in each category in $\Omega$. We now proceed to define this formally.

A {\em system of monoidal categories} $\Omega$ is a subcategory of $\Cat$ such that
\begin{itemize}[topsep=0pt,itemsep=-1ex,partopsep=1ex,parsep=1ex]
\item every category $\omega\in\Omega$ is strict monoidal,
\item every functor in $\Omega$ is strict monoidal,
\item there is at most one functor between any pair of categories, that is, $\Omega$ is posetal.
\end{itemize}
The last condition is assumed merely for simplicity, we could construct a layered prop from $\Omega$ with multiple monoidal functors between a pair of monoidal categories. The formalism could be modified to incorporate non-strict monoidal categories, we leave this for future work.

We view a system of monoidal categories as a system of sets in a straightforward manner: the collection of sets is given by $\{\Ob(\omega) : \omega\in\Ob(\Omega)\}$, we identify $\alpha\beta\coloneqq\alpha\otimes\beta$ for all $\alpha,\beta\in\omega$ and $\omega\in\Omega$, we have $\omega\leq\omega'$ whenever there is a functor $f:\omega'\rightarrow\omega$, and the monoid homomorphisms are given by the restriction of each functor to objetcs.

Assuming that all the categories $\omega\in\Omega$ as well as $\Omega$ itself are small, we define the signature $\Sigma(\Omega)$ as follows:
$$\Sigma(\Omega)\coloneqq \left\{\sigma_{\omega}^{\alpha,\beta}\right\}_{\omega\in \Ob(\Omega),\alpha,\beta\in\Ob(\omega),\sigma:\alpha\rightarrow\beta}.$$
The arities and coarities are defined as:
\begin{align*}
\arity\left(\sigma_{\omega}^{\alpha,\beta}\right) &= \omega,\alpha &  \coarity\left(\sigma_{\omega}^{\alpha,\beta}\right) &= \omega,\beta.
\end{align*}
In other words, $\Sigma(\Omega)$ contains every morphism in every category $\omega\in\Omega$.

\begin{definition}[Layered prop generated by a system of monoidal categories]
  A {\em layered prop generated by a system of monoidal categories} $\Omega$ is the layered prop generated by the layered monoidal theory $(\Omega,\Sigma(\Omega),\arity,\coarity)$. Additional generators for $2$-cells are given by the equalities of morphisms in each category $\omega\in\Omega$.
\end{definition}
We denote the layered prop generated by a system of monoidal categories $\Omega$ by $\mathcal L(\Omega)$.

\mathversion{normal3}
\section{Pointed Profunctor Semantics}\label{sec:semantics}
While our framework is purely syntactic (indeed, the whole point of constructing a layered prop is that we are able to treat all the layers in the same language), we are able to provide a semantic justification for the layered prop formalism: as we show in this section, they can be naturally interpreted in the category of pointed profunctors $\Prof_*$. We include the Appendix~\ref{sec:profunctors} on profunctors and pointed profunctors as a quick reference and to disambiguate any notation. For a proper introduction, see Borceux~\cite{borceux-bicategories} and Loregian~\cite{loregian}, and references therein.

\begin{definition}\label{def:semantics}
  Let $\mathcal L$ be a layered prop. A {\em profunctor model} of $\mathcal L$ is a $2$-functor $\mathcal L\rightarrow\Prof_*$ which is consistent in the sense that
  \begin{itemize}[topsep=0pt,itemsep=-1ex,partopsep=1ex,parsep=1ex]
    \item if the $0$-cells $(\omega,\alpha)$ and $(\omega,\beta)$ are respectively mapped to $(\cat C,c)$ and $(\cat D,d)$, then $\cat C=\cat D$,
    \item if the $1$-cells \scalebox{.6}{\begin{tikzpicture}
	\begin{pgfonlayer}{nodelayer}
		\node [style=none] (0) at (-2, 1) {};
		\node [style=none] (1) at (-2, -1) {};
		\node [style=none] (2) at (2, 1) {};
		\node [style=none] (3) at (2, -1) {};
		\node [style=none] (4) at (-2, 0) {};
		\node [style=none] (5) at (2, 0) {};
		\node [style=redbox] (6) at (0, 0) {$\sigma$};
		\node [style=none] (7) at (-2.75, 0) {$\alpha$};
		\node [style=none] (8) at (-3.5, 0) {$\omega$};
		\node [style=none] (9) at (2.75, 0) {$\beta$};
		\node [style=none] (10) at (3.5, 0) {$\omega$};
	\end{pgfonlayer}
	\begin{pgfonlayer}{edgelayer}
		\draw [style=rededge] (4.center) to (6);
		\draw [style=rededge] (6) to (5.center);
		\draw [style=edge] (0.center) to (1.center);
		\draw [style=edge] (1.center) to (3.center);
		\draw [style=edge] (3.center) to (2.center);
		\draw [style=edge] (2.center) to (0.center);
	\end{pgfonlayer}
\end{tikzpicture}} and \scalebox{.6}{\begin{tikzpicture}
	\begin{pgfonlayer}{nodelayer}
		\node [style=none] (0) at (-2, 1) {};
		\node [style=none] (1) at (-2, -1) {};
		\node [style=none] (2) at (2, 1) {};
		\node [style=none] (3) at (2, -1) {};
		\node [style=none] (4) at (-2, 0) {};
		\node [style=none] (5) at (2, 0) {};
		\node [style=redbox] (6) at (0, 0) {$\sigma'$};
		\node [style=none] (7) at (-2.75, 0) {$\alpha$};
		\node [style=none] (8) at (-3.5, 0) {$\omega$};
		\node [style=none] (9) at (2.75, 0) {$\beta$};
		\node [style=none] (10) at (3.5, 0) {$\omega$};
	\end{pgfonlayer}
	\begin{pgfonlayer}{edgelayer}
		\draw [style=rededge] (4.center) to (6);
		\draw [style=rededge] (6) to (5.center);
		\draw [style=edge] (0.center) to (1.center);
		\draw [style=edge] (1.center) to (3.center);
		\draw [style=edge] (3.center) to (2.center);
		\draw [style=edge] (2.center) to (0.center);
	\end{pgfonlayer}
\end{tikzpicture}} are respectively mapped to $(P,f)$ and $(Q,g)$, then $P=Q$.
  \end{itemize}
\end{definition}

For the rest of this section, we assume that $\Omega$ is a system of monoidal categories. We will show that there is a natural profunctor model of the layered prop generated by $\Omega$. To this end, we wish to define a $2$-functor $\mathcal I:\mathcal L(\Omega)\rightarrow\Prof_*$.

Let us define $\mathcal I$ on objects (i.e.~$\Omega$-types) recursively as follows: $\mathcal I(\varepsilon)\coloneqq (\one,\bullet)$, and $\mathcal I(t; \omega,\alpha)\coloneqq \mathcal I(t)\times (\omega,\alpha)$. In order to define $\mathcal I$ on morphisms, for each $\omega\in\Omega$, let us write $I_{\omega} : \one\rightarrow\omega$ for the functor sending the unique object of $\one$ to the monoidal unit of $\omega$. Likewise, let us write $\mathfrak s : \cat C\times\cat D\rightarrow\cat D\times\cat C$ for the symmetry map in $\Cat$. Note that since $\mathfrak s$ is an isomorphism, we have $\embeddn {\mathfrak s}\simeq\embedup {\mathfrak s}$. We then define $\mathcal I$ by the following action on the generators:
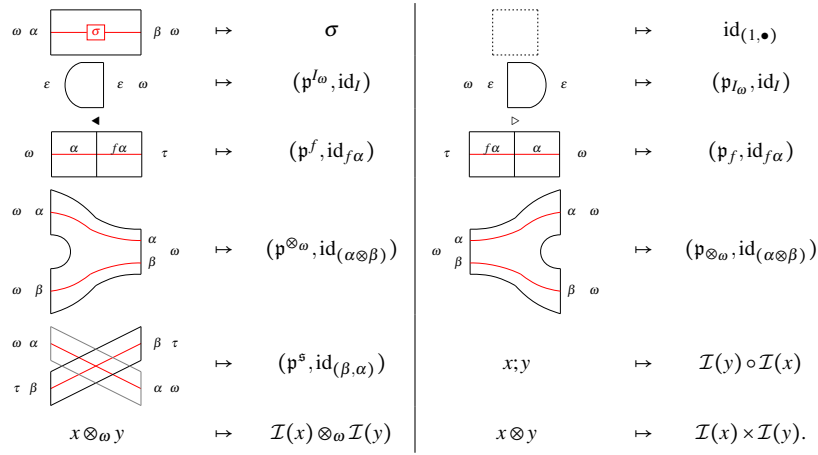
\bgroup
\def\arraystretch{1.85} 
\begin{center}
\scriptsize
\begin{tabular}{ c c c | c c c }
\scalebox{.6}{\begin{tikzpicture}
	\begin{pgfonlayer}{nodelayer}
		\node [style=none] (0) at (-2, 1) {};
		\node [style=none] (1) at (-2, -1) {};
		\node [style=none] (2) at (2, 1) {};
		\node [style=none] (3) at (2, -1) {};
		\node [style=none] (4) at (-2, 0) {};
		\node [style=none] (5) at (2, 0) {};
		\node [style=redbox] (6) at (0, 0) {$\sigma$};
		\node [style=none] (7) at (-2.75, 0) {$\alpha$};
		\node [style=none] (8) at (-3.5, 0) {$\omega$};
		\node [style=none] (9) at (2.75, 0) {$\beta$};
		\node [style=none] (10) at (3.5, 0) {$\omega$};
	\end{pgfonlayer}
	\begin{pgfonlayer}{edgelayer}
		\draw [style=rededge] (4.center) to (6);
		\draw [style=rededge] (6) to (5.center);
		\draw [style=edge] (0.center) to (1.center);
		\draw [style=edge] (1.center) to (3.center);
		\draw [style=edge] (3.center) to (2.center);
		\draw [style=edge] (2.center) to (0.center);
	\end{pgfonlayer}
\end{tikzpicture}} & $\mapsto$ & $\sigma$ & \scalebox{.6}{\begin{tikzpicture}
	\begin{pgfonlayer}{nodelayer}
		\node [style=none] (0) at (-1, -1) {};
		\node [style=none] (1) at (-1, 1) {};
		\node [style=none] (2) at (1, 1) {};
		\node [style=none] (3) at (1, -1) {};
	\end{pgfonlayer}
	\begin{pgfonlayer}{edgelayer}
		\draw [style=dashed edge] (1.center) to (0.center);
		\draw [style=dashed edge] (0.center) to (3.center);
		\draw [style=dashed edge] (3.center) to (2.center);
		\draw [style=dashed edge] (2.center) to (1.center);
	\end{pgfonlayer}
\end{tikzpicture}} & $\mapsto$ & $\id_{(\one,\bullet)}$ \\
\scalebox{.6}{\begin{tikzpicture}
	\begin{pgfonlayer}{nodelayer}
		\node [style=none] (0) at (-0.25, 1) {};
		\node [style=none] (1) at (-0.25, -1) {};
		\node [style=none] (2) at (0.75, -1) {};
		\node [style=none] (3) at (0.75, 1) {};
		\node [style=none] (4) at (1.5, 0) {$\varepsilon$};
		\node [style=none] (5) at (2.5, 0) {$\omega$};
		\node [style=none] (6) at (-1.75, 0) {$\varepsilon$};
	\end{pgfonlayer}
	\begin{pgfonlayer}{edgelayer}
		\draw [style=edge, bend right=75, looseness=1.25] (0.center) to (1.center);
		\draw [style=edge] (1.center) to (2.center);
		\draw [style=edge] (2.center) to (3.center);
		\draw [style=edge] (3.center) to (0.center);
	\end{pgfonlayer}
\end{tikzpicture}} & $\mapsto$ & $(\embedup {I_{\omega}}, \id_I)$ & \scalebox{.6}{\begin{tikzpicture}
	\begin{pgfonlayer}{nodelayer}
		\node [style=none] (0) at (0.5, 1) {};
		\node [style=none] (1) at (0.5, -1) {};
		\node [style=none] (2) at (-0.5, -1) {};
		\node [style=none] (3) at (-0.5, 1) {};
		\node [style=none] (4) at (-1.25, 0) {$\varepsilon$};
		\node [style=none] (5) at (-2.25, 0) {$\omega$};
		\node [style=none] (6) at (2, 0) {$\varepsilon$};
	\end{pgfonlayer}
	\begin{pgfonlayer}{edgelayer}
		\draw [style=edge, bend left=75, looseness=1.25] (0.center) to (1.center);
		\draw [style=edge] (1.center) to (2.center);
		\draw [style=edge] (2.center) to (3.center);
		\draw [style=edge] (3.center) to (0.center);
	\end{pgfonlayer}
\end{tikzpicture}} & $\mapsto$ & $(\embeddn {I_{\omega}}, \id_I)$ \\
\scalebox{.6}{\begin{tikzpicture}
	\begin{pgfonlayer}{nodelayer}
		\node [style=none] (0) at (-3, 0) {$\omega$};
		\node [style=none] (1) at (-2, 1) {};
		\node [style=none] (2) at (-2, -1) {};
		\node [style=none] (3) at (0, 1) {};
		\node [style=none] (4) at (0, -1) {};
		\node [style=none] (5) at (2, 1) {};
		\node [style=none] (6) at (2, -1) {};
		\node [style=none] (7) at (3, 0) {$\tau$};
		\node [style=none] (8) at (0, 1.5) {$\refine$};
		\node [style=none] (9) at (-2, 0) {};
		\node [style=none] (10) at (2, 0) {};
		\node [style=none] (11) at (-1, 0.25) {$\alpha$};
		\node [style=none] (12) at (1, 0.25) {$f\alpha$};
		\node [style=none] (13) at (0, 2) {};
	\end{pgfonlayer}
	\begin{pgfonlayer}{edgelayer}
		\draw [style=rededge] (9.center) to (10.center);
		\draw [style=edge] (3.center) to (4.center);
		\draw [style=edge] (1.center) to (5.center);
		\draw [style=edge] (1.center) to (2.center);
		\draw [style=edge] (2.center) to (6.center);
		\draw [style=edge] (6.center) to (5.center);
	\end{pgfonlayer}
\end{tikzpicture}} & $\mapsto$ & $(\embedup f, \id_{f\alpha})$ & \scalebox{.6}{\begin{tikzpicture}
	\begin{pgfonlayer}{nodelayer}
		\node [style=none] (0) at (3, 0) {$\omega$};
		\node [style=none] (1) at (2, 1) {};
		\node [style=none] (2) at (2, -1) {};
		\node [style=none] (3) at (0, 1) {};
		\node [style=none] (4) at (0, -1) {};
		\node [style=none] (5) at (-2, 1) {};
		\node [style=none] (6) at (-2, -1) {};
		\node [style=none] (7) at (-3, 0) {$\tau$};
		\node [style=none] (8) at (0, 1.5) {$\coarsen$};
		\node [style=none] (9) at (2, 0) {};
		\node [style=none] (10) at (-2, 0) {};
		\node [style=none] (11) at (1, 0.25) {$\alpha$};
		\node [style=none] (12) at (-1, 0.25) {$f\alpha$};
		\node [style=none] (13) at (0, 2) {};
	\end{pgfonlayer}
	\begin{pgfonlayer}{edgelayer}
		\draw [style=rededge] (9.center) to (10.center);
		\draw [style=edge] (3.center) to (4.center);
		\draw [style=edge] (5.center) to (1.center);
		\draw [style=edge] (1.center) to (2.center);
		\draw [style=edge] (2.center) to (6.center);
		\draw [style=edge] (6.center) to (5.center);
	\end{pgfonlayer}
\end{tikzpicture}} & $\mapsto$ & $(\embeddn f, \id_{f\alpha})$ \\
\scalebox{.6}{\begin{tikzpicture}
	\begin{pgfonlayer}{nodelayer}
		\node [style=none] (0) at (-1.5, 0.75) {};
		\node [style=none] (1) at (-1.5, 2.75) {};
		\node [style=none] (2) at (-1.5, -0.75) {};
		\node [style=none] (3) at (-1.5, -2.75) {};
		\node [style=none] (4) at (2.5, 1) {};
		\node [style=none] (5) at (2.5, -1) {};
		\node [style=none] (6) at (0.5, 1.5) {};
		\node [style=none] (7) at (0.5, -1.5) {};
		\node [style=none] (8) at (-1.5, 1.75) {};
		\node [style=none] (10) at (2.5, 0.5) {};
		\node [style=none] (12) at (0.25, 1) {};
		\node [style=none] (13) at (-1.5, -1.75) {};
		\node [style=none] (14) at (2.5, -0.5) {};
		\node [style=none] (15) at (0.25, -1) {};
		\node [style=none] (16) at (-3, 1.75) {$\omega$};
		\node [style=none] (17) at (-3, -1.75) {$\omega$};
		\node [style=none] (18) at (4, 0) {$\omega$};
		\node [style=none] (19) at (-2, 1.75) {};
		\node [style=none] (20) at (-2, 1.75) {$\alpha$};
		\node [style=none] (21) at (3, 0.5) {$\alpha$};
		\node [style=none] (22) at (3, -0.5) {$\beta$};
		\node [style=none] (23) at (-2, -1.75) {$\beta$};
		\node [style=none] (24) at (0, 3.25) {};
	\end{pgfonlayer}
	\begin{pgfonlayer}{edgelayer}
		\draw [style=edge] (5.center) to (4.center);
		\draw [style=edge] (1.center) to (0.center);
		\draw [style=edge] (3.center) to (2.center);
		\draw [style=edge, bend left=90, looseness=2.00] (0.center) to (2.center);
		\draw [style=edge, bend left=15] (1.center) to (6.center);
		\draw [style=edge, bend right=15] (6.center) to (4.center);
		\draw [style=edge, bend right=15] (3.center) to (7.center);
		\draw [style=edge, bend left=15] (7.center) to (5.center);
		\draw [style=rededge, bend left=15] (8.center) to (12.center);
		\draw [style=rededge, bend right=15, looseness=0.75] (12.center) to (10.center);
		\draw [style=rededge, bend right=15] (13.center) to (15.center);
		\draw [style=rededge, bend left=15, looseness=0.75] (15.center) to (14.center);
	\end{pgfonlayer}
\end{tikzpicture}} & $\mapsto$ & $(\embedup {\otimes_{\omega}}, \id_{(\alpha\otimes\beta)})$ & \scalebox{.6}{\begin{tikzpicture}
	\begin{pgfonlayer}{nodelayer}
		\node [style=none] (0) at (2.5, 0.75) {};
		\node [style=none] (1) at (2.5, 2.75) {};
		\node [style=none] (2) at (2.5, -0.75) {};
		\node [style=none] (3) at (2.5, -2.75) {};
		\node [style=none] (4) at (-1.5, 1) {};
		\node [style=none] (5) at (-1.5, -1) {};
		\node [style=none] (6) at (0.5, 1.5) {};
		\node [style=none] (7) at (0.5, -1.5) {};
		\node [style=none] (8) at (2.5, 1.75) {};
		\node [style=none] (10) at (-1.5, 0.5) {};
		\node [style=none] (12) at (0.75, 1) {};
		\node [style=none] (13) at (2.5, -1.75) {};
		\node [style=none] (14) at (-1.5, -0.5) {};
		\node [style=none] (15) at (0.75, -1) {};
		\node [style=none] (16) at (4, 1.75) {$\omega$};
		\node [style=none] (17) at (4, -1.75) {$\omega$};
		\node [style=none] (18) at (-3, 0) {$\omega$};
		\node [style=none] (19) at (3, 1.75) {};
		\node [style=none] (20) at (3, 1.75) {$\alpha$};
		\node [style=none] (21) at (-2, 0.5) {$\alpha$};
		\node [style=none] (22) at (-2, -0.5) {$\beta$};
		\node [style=none] (23) at (3, -1.75) {$\beta$};
		\node [style=none] (24) at (0, 3.25) {};
	\end{pgfonlayer}
	\begin{pgfonlayer}{edgelayer}
		\draw [style=edge] (5.center) to (4.center);
		\draw [style=edge] (1.center) to (0.center);
		\draw [style=edge] (3.center) to (2.center);
		\draw [style=edge, bend right=90, looseness=2.00] (0.center) to (2.center);
		\draw [style=edge, bend right=15] (1.center) to (6.center);
		\draw [style=edge, bend left=15] (6.center) to (4.center);
		\draw [style=edge, bend left=15] (3.center) to (7.center);
		\draw [style=edge, bend right=15] (7.center) to (5.center);
		\draw [style=rededge, bend right=15] (8.center) to (12.center);
		\draw [style=rededge, bend left=15, looseness=0.75] (12.center) to (10.center);
		\draw [style=rededge, bend left=15] (13.center) to (15.center);
		\draw [style=rededge, bend right=15, looseness=0.75] (15.center) to (14.center);
	\end{pgfonlayer}
\end{tikzpicture}} & $\mapsto$ & $(\embeddn {\otimes_{\omega}}, \id_{(\alpha\otimes\beta)})$ \\
\scalebox{.6}{\begin{tikzpicture}
	\begin{pgfonlayer}{nodelayer}
		\node [style=none] (6) at (-2, 1) {};
		\node [style=none] (7) at (2, -1) {};
		\node [style=none] (0) at (-2, 1.75) {};
		\node [style=none] (1) at (-2, 0.25) {};
		\node [style=none] (2) at (2, -1.75) {};
		\node [style=none] (3) at (2, -0.25) {};
		\node [style=none] (4) at (3.5, -1) {$\omega$};
		\node [style=none] (5) at (-3.5, 1) {$\omega$};
		\node [style=none] (8) at (-2.75, 1) {$\alpha$};
		\node [style=none] (9) at (2, 1) {};
		\node [style=none] (10) at (-2, -1) {};
		\node [style=none] (11) at (2, 1.75) {};
		\node [style=none] (12) at (2, 0.25) {};
		\node [style=none] (13) at (-2, -1.75) {};
		\node [style=none] (14) at (-2, -0.25) {};
		\node [style=none] (15) at (-3.5, -1) {$\tau$};
		\node [style=none] (16) at (3.5, 1) {$\tau$};
		\node [style=none] (17) at (-2.75, -1) {$\beta$};
		\node [style=none] (18) at (2.75, 1) {$\beta$};
		\node [style=none] (19) at (2.75, -1) {$\alpha$};
		\node [style=none] (20) at (0, 2.25) {};
	\end{pgfonlayer}
	\begin{pgfonlayer}{edgelayer}
		\draw [style=rededge] (6.center) to (7.center);
		\draw [style=rededge] (9.center) to (10.center);
		\draw [style=edge] (11.center) to (12.center);
		\draw [style=edge] (12.center) to (13.center);
		\draw [style=edge] (13.center) to (14.center);
		\draw [style=edge] (14.center) to (11.center);
		\draw [style=gray edge] (0.center) to (1.center);
		\draw [style=gray edge] (0.center) to (3.center);
		\draw [style=gray edge] (1.center) to (2.center);
		\draw [style=gray edge] (3.center) to (2.center);
	\end{pgfonlayer}
\end{tikzpicture}} & $\mapsto$ & $(\embedup {\mathfrak s}, \id_{(\beta,\alpha)})$ & $x;y$ & $\mapsto$ & $\mathcal I(y)\circ\mathcal I(x)$ \\
$x\otimes_{\omega} y$ & $\mapsto$ & $\mathcal I(x)\otimes_{\omega}\mathcal I(y)$ & $x\otimes y$ & $\mapsto$ & $\mathcal I(x)\times\mathcal I(y)$.
\end{tabular}
\end{center}
\egroup
where $\embedup -$ and $\embeddn -$ are the covariant and contravariant embeddings of $\Cat$ in the category of profunctors, and $\sigma$ stands for the pointed profunctor $(\hom_{\omega},\sigma)$. We prove the following proposition in Appendix~\ref{sec:proofs}.
\begin{proposition}\label{prop:monoidal-prof-model}
The assignment $\mathcal I$ is a profunctor model of $\mathcal L(\Omega)$. Namely, it preserves the equalities of morphisms in each category $\omega\in\Omega$ as well as the rules in Figures~\ref{fig:axioms-pants1}, \ref{fig:axioms-pants2} and~\ref{fig:axioms-pants3}.
\end{proposition}

\section{Explanations}\label{sec:explanations}

Using the formalism introduced in the previous sections, we are now able to formulate precisely the notion of an explanation. First, we give names to two special shapes of $1$-cells in a layered prop and outline their connection to explanations. We assume that we are working with a layered prop generated by a system of monoidal categories $\Omega$.

\begin{definition}[Window, cowindow]\label{def:window}
  A {\em window} is a morphism in a layered prop of the form on the left below.   Dually, a {\em cowindow} is a morphism in a layered prop of the form on the right below.
  \begin{center}
  \scalebox{.6}
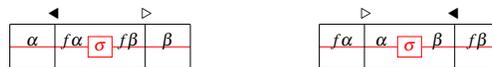
{\begin{tikzpicture}
	\begin{pgfonlayer}{nodelayer}
		\node [style=none] (0) at (-4, 1) {};
		\node [style=none] (1) at (-4, -1) {};
		\node [style=none] (2) at (-2, 1) {};
		\node [style=none] (3) at (2, 1) {};
		\node [style=none] (4) at (4, 1) {};
		\node [style=none] (5) at (4, -1) {};
		\node [style=none] (8) at (-2, -1) {};
		\node [style=none] (9) at (2, -1) {};
		\node [style=none] (10) at (-2, 1.5) {$\refine$};
		\node [style=none] (11) at (2, 1.5) {$\coarsen$};
		\node [style=redbox] (12) at (0, 0) {$\sigma$};
		\node [style=none] (13) at (-4, 0) {};
		\node [style=none] (14) at (4, 0) {};
		\node [style=none] (15) at (-3, 0.25) {$\alpha$};
		\node [style=none] (16) at (-1.25, 0.25) {$f\alpha$};
		\node [style=none] (17) at (1.25, 0.25) {$f\beta$};
		\node [style=none] (18) at (3, 0.25) {$\beta$};
	\end{pgfonlayer}
	\begin{pgfonlayer}{edgelayer}
		\draw [style=rededge] (13.center) to (12);
		\draw [style=rededge] (12) to (14.center);
		\draw [style=edge] (2.center) to (8.center);
		\draw [style=edge] (3.center) to (9.center);
		\draw [style=edge] (0.center) to (4.center);
		\draw [style=edge] (4.center) to (5.center);
		\draw [style=edge] (5.center) to (1.center);
		\draw [style=edge] (1.center) to (0.center);
	\end{pgfonlayer}
\end{tikzpicture}} \qquad \qquad   \scalebox{.6}{\begin{tikzpicture}
	\begin{pgfonlayer}{nodelayer}
		\node [style=none] (0) at (-4, 1) {};
		\node [style=none] (1) at (-4, -1) {};
		\node [style=none] (2) at (-2, 1) {};
		\node [style=none] (3) at (2, 1) {};
		\node [style=none] (4) at (4, 1) {};
		\node [style=none] (5) at (4, -1) {};
		\node [style=none] (8) at (-2, -1) {};
		\node [style=none] (9) at (2, -1) {};
		\node [style=none] (10) at (2, 1.5) {$\refine$};
		\node [style=none] (11) at (-2, 1.5) {$\coarsen$};
		\node [style=redbox] (12) at (0, 0) {$\sigma$};
		\node [style=none] (13) at (-4, 0) {};
		\node [style=none] (14) at (4, 0) {};
		\node [style=none] (15) at (-3, 0.25) {$f\alpha$};
		\node [style=none] (16) at (-1.25, 0.25) {$\alpha$};
		\node [style=none] (17) at (1.25, 0.25) {$\beta$};
		\node [style=none] (18) at (3, 0.25) {$f\beta$};
	\end{pgfonlayer}
	\begin{pgfonlayer}{edgelayer}
		\draw [style=rededge] (13.center) to (12);
		\draw [style=rededge] (12) to (14.center);
		\draw [style=edge] (2.center) to (8.center);
		\draw [style=edge] (3.center) to (9.center);
		\draw [style=edge] (0.center) to (4.center);
		\draw [style=edge] (4.center) to (5.center);
		\draw [style=edge] (5.center) to (1.center);
		\draw [style=edge] (1.center) to (0.center);
	\end{pgfonlayer}
\end{tikzpicture}}
  \end{center}
\end{definition}

Windows correspond to {\em reductive} explanations: a process at the higher level gets translated to the lower level, where we can apply laws or rules that are (presumably) more flexible, after which we translate back to the higher level, hence completing the explanation. This remark should be compared to the shape of the explanation of glucose phosphorylation in Section~\ref{sec:glucose} below.

Cowindows, in turn, correspond to {\em functional} explanations: a process at the lower level is justified by passing through a higher level in such a way that the higher level process translates back to what is being explained. It can thus be thought that the lower level process takes place in order to yield the appropriate form at the higher level. Axioms \texttt{unit-ref} and \texttt{counit-ref} state that there is an asymmetry between reductive and functional explanations: using \texttt{unit-ref}, it is always possible to create a window (and hence give a reductive explanation), while \texttt{counit-ref} only allows reducing a trivial cowindow to the identity. Note that what we call a cowindow is usually called a {\em functorial box} in the literature --- see e.g.~\cite{functorial-boxes}.

We may now define what an explanation is: we do this separately for $1$-cells and for $2$-cells. Both correspond to a {\em reduction}: an explanation of a $1$-cell reduces a process to another one at a lower level of abstraction, while an explanation of a $2$-cell reduces a rule between two processes to a rule between reductions of these processes. For examples of explanations (now in a formal sense), see Figure~\ref{fig:glucose} and the discussion in Section~\ref{sec:concurrency} ($1$-cell), and Figure~\ref{fig:resistors} ($2$-cell).

\begin{definition}[Explanation of a $1$-cell]\label{def:explanation-1}
  Let $\mathfrak e$ and $\sigma$ be parallel $1$-cells in a layered prop (that is, having the same domain and codomain). We say that $\mathfrak e$ is an {\em explanation} of $\sigma$ if
  \begin{enumerate}[topsep=0pt,itemsep=-1ex,partopsep=1ex,parsep=1ex]
    \item $\sigma$ is an internal morphism contained in some category $\omega\in\Omega$,
    \item every internal non-identity morphism of $\mathfrak e$ is contained in some category $\omega'$ such that $\omega' < \omega$ in the partial order of $\Omega$,
    \item there is either a $2$-cell $\mathfrak e\rightarrow\sigma$ or a $2$-cell $\sigma\rightarrow\mathfrak e$.
  \end{enumerate}
\end{definition}

\begin{definition}[Explanation of a $2$-cell]\label{def:explanation-2}
  Let $\eta$ and $\mu$ be parallel $2$-cells in a layered prop. We say that $\eta$ is an {\em explanation} of $\mu$ if
  \begin{enumerate}[topsep=0pt,itemsep=-1ex,partopsep=1ex,parsep=1ex]
    \item $\mu$ is generated by an equality of morphisms in some category $\omega\in\Omega$,
    \item $\eta$ can be constructed using the generating $2$-cells of a layered prop and the $2$-cells that come from an equality of morphisms in those categories $\omega'$ for which $\omega' < \omega$ in the partial order of $\Omega$.
  \end{enumerate}
\end{definition}

The above definitions correspond to the intuitive understanding of a (reductive) explanation we outlined in Section~\ref{sec:introduction}. The first condition in both definitions ensures that what is being explained is internal to a particular category, that is, to a description at a particular level of abstraction. The second condition says that the explanation is indeed reductive: it may only use lower levels of description than what is being explained (in addition to the metalanguage of the layered prop). This implies that an explanation must contain at least one window. The third condition in the first definition ensures that the explanation is {\em relevant} in the sense that it is either a sufficient or a necessary cause for what is being explained. There is no such condition for the explanations of $2$-cells since in our setup there are no $3$-cells. We thus simply assume that an explanation is relevant. This assumption need not be made if we are working with higher categories. These definitions can be dualised, this gives definitions of functional explanations, or ``coexplanations''. We will not need these in this work, and therefore omit the explicit statements.

Interestingly, if we require that the third condition of Definition~\ref{def:explanation-1} does not hold (i.e.~there is no $2$-cell between the explanation and what is being explained), we obtain the definition of a {\em counterfactual explanation}. Ability to model counterfactual reasoning is important for the causal analysis in the rule-based models of molecular interactions, such as the Kappa language~\cite{counterfactual}. While a particular simulation of a rule-based model may tell us that a rule $\mathfrak e$ was invoked in the computation of the effect $\sigma$, so that $\mathfrak e$ explains $\sigma$ in the sense of Definition~\ref{def:explanation-1}, this tells nothing about necessity (or sufficiency) of $\mathfrak e$ for $\sigma$. Thus a rule-based model is not (without modifications) able to deal with such questions as {\em Would $\sigma$ occur had $\mathfrak e$ not occurred?} In a layered prop, the positive answer to such question (establishing non-necessity) can be provided by finding a counterfactual explanation of $\sigma$ that has the same sort as $\mathfrak e$. Intuitively, a (possibly) counterfactual explanation can be thought of as a $1$-cell that ``fills in the gap'' left by $\mathfrak e$:
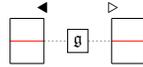
\begin{center}
\scalebox{.6}{\begin{tikzpicture}
	\begin{pgfonlayer}{nodelayer}
		\node [style=none] (0) at (-3, 1) {};
		\node [style=none] (1) at (-3, -1) {};
		\node [style=none] (2) at (-1.5, 1) {};
		\node [style=none] (3) at (-1.5, -1) {};
		\node [style=none] (4) at (1.5, 1) {};
		\node [style=none] (5) at (1.5, -1) {};
		\node [style=none] (6) at (3, 1) {};
		\node [style=none] (7) at (3, -1) {};
		\node [style=none] (8) at (-3, 0) {};
		\node [style=none] (9) at (-1.5, 0) {};
		\node [style=none] (10) at (1.5, 0) {};
		\node [style=none] (11) at (3, 0) {};
		\node [style=bdytextbox] (12) at (0, 0) {$\mathfrak g$};
		\node [style=none] (13) at (-1.5, 1.5) {$\refine$};
		\node [style=none] (14) at (1.5, 1.5) {$\coarsen$};
	\end{pgfonlayer}
	\begin{pgfonlayer}{edgelayer}
		\draw [style=rededge] (8.center) to (9.center);
		\draw [style=rededge] (10.center) to (11.center);
		\draw [style=gray dashed] (9.center) to (12);
		\draw [style=gray dashed] (12) to (10.center);
		\draw [style=edge] (2.center) to (0.center);
		\draw [style=edge] (0.center) to (1.center);
		\draw [style=edge] (3.center) to (1.center);
		\draw [style=edge] (3.center) to (2.center);
		\draw [style=edge] (4.center) to (6.center);
		\draw [style=edge] (6.center) to (7.center);
		\draw [style=edge] (7.center) to (5.center);
		\draw [style=edge] (5.center) to (4.center);
	\end{pgfonlayer}
\end{tikzpicture}}
\end{center}
We give an example of a counterfactual explanation in our discussion of concurrency in Section~\ref{sec:concurrency}.

Models based on variable substitution~\cite{pearl-causality} and trajectory sampling~\cite{counterfactual} have been proposed to model counterfactual statements. Since our setup remains agnostic about what the internal morphisms in a layered prop actually are, we expect that both of these situations can be modelled within a layered prop. We leave this investigation for future work.

\mathversion{normal4}
\section{Example: Glucose Phosphorylation}\label{sec:glucose}

In this section, we construct a minimal example --- inspired by Krivine~\cite{krivine-talk} --- that illustrates our notion of an explanation (specifically, Definition~\ref{def:explanation-1}) for an important biochemical process known as {\em phosphorylation of glucose}. This is motivated by the problem of systematising a vast amount of experimental data in systems biology in a way that is easy for humans to both understand and use. Our strategy is to define three monoidal categories that are capable of talking about chemical reactions at three different abstraction levels:
\begin{center}
\begin{tabular}{ c | c }
$\Lplus$ & English names of the relevant molecules \\
\hline
$\Molp$ & Molecules \\
\hline
$\PartMolp$ & Partitions of molecules into smaller units
\end{tabular}
\end{center}

First, let us define $\Lplus$ as the free monoidal category with generating objects
$$\{\glucose, \ATP, \glucosesix, \ADP, \hydrogenion\},$$
whose monoidal product is denoted by $+$, and with just one generating morphism
\begin{equation}\label{eqn:phosphorylation}
\glucose+\ATP \longrightarrow \glucosesix+\ADP+\hydrogenion.
\end{equation}
The generating morphism simply represents the high-level chemical rule describing phosphorylation of glucose. Here $\ATP$ and $\ADP$ stand for {\em adenosine triphosphate} and {\em adenosine diphosphate}.

\subsection{Molecules and Molecule Partitions}
We define a {\em molecule partition} as a certain connected multigraph (Definition~\ref{def:molpart}). We then identify as {\em molecules} those molecule partitions that do not have free variables. Fix a countable set of {\em free variables} $\FW$. We denote the elements of $\FW$ by lowercase Greek letters $\alpha,\beta,\gamma,\dots$. Let us define the set of {\em atoms} as containing the symbol for each main-group element of the periodic table together with the symbols $-$ and $+$: $\At\coloneqq\{-,+,H,C,O,P,\dots\}$. Define the function $\mathbf v:\At\sqcup\FW\rightarrow\N$ as taking each element symbol to the valence of that element\footnote{This is a bit of a naive model, as valence is in general context-sensitive and not determined by a single atom. Yet this is good enough for the purposes of this example.}, define $\mathbf v(-)=\mathbf v(+)=1$ and finally for all $\alpha\in\FW$ let $\mathbf v(\alpha)=1$.

\begin{definition}[Molecule partition]\label{def:molpart}
A {\em molecule partition} is a triple $(V,\tau,m)$, where $V$ is a finite set of {\em vertices}, $\tau:V\rightarrow\At\sqcup\FW$ is a function taking each vertex to its {\em type} and $m:V\times V\rightarrow\N$ is a function satisfying the following conditions:
\begin{itemize}[topsep=0pt,itemsep=-1ex,partopsep=1ex,parsep=1ex]
\item for all $v\in V$, we have $m(v,v)=0$,
\item for all $v,w\in V$, we have $m(v,w)=m(w,v)$,
\item for all $v,u\in V$ with $v\neq u$, there are $w_0,\dots,w_n\in V$ such that $w_0=v$ and $w_n=u$ and $m(w_{i-1},w_i)\neq 0$ for each $i=1,\dots,n$,
\item for all $v\in V$, we have $\sum_{u\in V}m(u,v)=\mathbf v\tau(v)$.
\end{itemize}
In other words, the integers $m(i,j)$ form an adjacency matrix of an irreflexive, symmetric and connected multigraph, and the sum of each row or column gives the valence of the (type of) corresponding vertex.
\end{definition}

\begin{definition}[Molecule]
We say that a molecule partition $(V,\tau,m)$ is a {\em molecule} if the image of the function $\tau$ is contained in $\At$.
\end{definition}

We denote the set of all molecules by $\Mol$ and the set of all molecule partitions by $\PartMol$. Define the {\em partitioning relation} $R\sse\PartMol\times(\PartMol\times\PartMol)$ as follows. Let $M=(V,\tau,m)$ be a molecule partition, let $u,v\in V$ and let $\alpha\in\FW$. Denote by $m':V\times V\rightarrow\N$ the function such that $m'(u,v)=m'(v,u)=0$ and $m'=m$ otherwise. Suppose that the following conditions are satisfied:
\begin{enumerate}[topsep=0pt,itemsep=-1ex,partopsep=1ex,parsep=1ex]
\item $m(u,v)=1$,
\item the graph $(V,m')$ is not connected,
\item $\alpha$ does not appear as a free variable in $M$ (that is, $\tau(w)\neq\alpha$ for all $w\in V$).
\end{enumerate}
In such case we denote by $V(u)$ and $V(v)$ the connected components of $u$ and $v$, respectively, in $(V,m')$. Let $M_u^{\alpha}=(V(u)\sqcup\{\alpha\},\tau_{\alpha},m_u)$ be the molecule partition where $\tau_{\alpha}(\alpha)=\alpha$ and $\tau_{\alpha}=\tau$ otherwise, and $m_u(u,\alpha)=m_u(\alpha,u)=1$ and $m_u=m$ otherwise. The molecule partition $M_v^{\alpha}=(V(v)\sqcup\{\alpha\},\tau_{\alpha},m_v)$ is defined similarly. Now we finally define $R$ by stipulating that $MR(M_u^{\alpha},M_v^{\alpha})$ for all $M$, $v$, $u$ and $\alpha$ that satisfy the above conditions.

Let us define $\Molp$ as the free monoidal category with generating objects $\Mol$ and just one generating morphism, which has the same shape as the generating morphism of $\Lplus$~\eqref{eqn:phosphorylation}, except that all the English names of the molecules are translated to the corresponding graphs (see Figure~\ref{fig:translation}). Similarly, define $\PartMolp$ as the free monoidal category with generating objects $\PartMol$. For all variables $\alpha$ and $\beta$ we add the rule
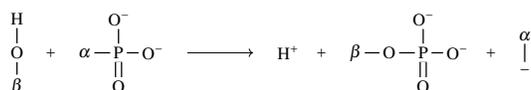
\begin{center}
\scalebox{.6}{\begin{tikzpicture}
	\begin{pgfonlayer}{nodelayer}
		\node [style=none] (17) at (-2.25, 0) {};
		\node [style=none] (18) at (0.75, 0) {};
		\node [style=none] (19) at (-8.25, 0) {$+$};
		\node [style=none] (20) at (3.75, 0) {$+$};
		\node [style=none] (21) at (11.25, 0) {$+$};
		\node [style=textbox] (22) at (-9.75, -1.5) {$\beta$};
		\node [style=textbox] (23) at (-9.75, 0) {O};
		\node [style=textbox] (24) at (-9.75, 1.5) {H};
		\node [style=textbox] (25) at (-5.25, 0) {P};
		\node [style=textbox] (26) at (-5.25, -1.5) {O};
		\node [style=textbox] (27) at (-6.75, 0) {$\alpha$};
		\node [style=textbox] (28) at (-5.25, 1.5) {\supmin O};
		\node [style=textbox] (29) at (-3.75, 0) {\supmin O};
		\node [style=textbox] (30) at (2.25, 0) {\supplus H};
		\node [style=textbox] (31) at (8.25, 0) {P};
		\node [style=textbox] (32) at (8.25, -1.5) {O};
		\node [style=textbox] (33) at (6.75, 0) {O};
		\node [style=textbox] (34) at (8.25, 1.5) {\supmin O};
		\node [style=textbox] (35) at (9.75, 0) {\supmin O};
		\node [style=textbox] (36) at (5.25, 0) {$\beta$};
		\node [style=textbox] (37) at (12.75, 0.75) {$\alpha$};
		\node [style=textbox] (38) at (12.75, -0.75) {$-$};
	\end{pgfonlayer}
	\begin{pgfonlayer}{edgelayer}
		\draw [style=diredge] (17) to (18);
		\draw [style=thedge] (22) to (23);
		\draw [style=thedge] (23) to (24);
		\draw [style=thedge] (27) to (25);
		\draw [style=thedge] (25) to (28);
		\draw [style=thedge] (25) to (29);
		\draw [style=doubedge] (25) to (26);
		\draw [style=thedge] (33) to (31);
		\draw [style=thedge] (31) to (34);
		\draw [style=thedge] (31) to (35);
		\draw [style=doubedge] (31) to (32);
		\draw [style=thedge] (36) to (33);
		\draw [style=thedge] (37) to (38);
	\end{pgfonlayer}
\end{tikzpicture}}
\end{center}
as a generating morphism to $\PartMolp$. We draw this as a box: \scalebox{.6}{\begin{tikzpicture}
	\begin{pgfonlayer}{nodelayer}
		\node [style=redbox] (0) at (0, 0) {$A_{\PartMolp}$};
		\node [style=none] (1) at (-3.5, 0.25) {};
		\node [style=none] (2) at (-3.5, -0.25) {};
		\node [style=none] (3) at (3.5, 0.25) {};
		\node [style=none] (4) at (3.5, -0.25) {};
		\node [style=none] (6) at (3.5, 0) {};
	\end{pgfonlayer}
	\begin{pgfonlayer}{edgelayer}
		\draw [style=rededge] (1.center) to (3.center);
		\draw [style=rededge] (2.center) to (4.center);
		\draw [style=rededge] (0) to (6.center);
	\end{pgfonlayer}
\end{tikzpicture}}. Further, for all molecule partitions $M,N$ and $K$ such that $MR(N,K)$ we introduce the following generators
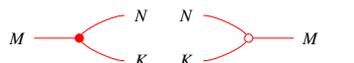
\begin{center}
\scalebox{.6}{\begin{tikzpicture}
	\begin{pgfonlayer}{nodelayer}
		\node [style=none] (8) at (1.75, 1) {};
		\node [style=none] (10) at (1.75, -1) {};
		\node [style=new style 0] (1) at (3.75, 0) {};
		\node [style=none] (2) at (-5.75, 0) {};
		\node [style=none] (3) at (-1.75, 1) {};
		\node [style=none] (5) at (5.75, 0) {};
		\node [style=rednode] (6) at (-3.75, 0) {};
		\node [style=none] (7) at (-1.75, -1) {};
		\node [style=none] (11) at (-6.5, 0) {$M$};
		\node [style=none] (12) at (-1, 1) {$N$};
		\node [style=none] (13) at (-1, -1) {$K$};
		\node [style=none] (14) at (6.5, 0) {$M$};
		\node [style=none] (15) at (1, 1) {$N$};
		\node [style=none] (16) at (1, -1) {$K$};
		\node [style=none] (17) at (8.5, 0) {.};
	\end{pgfonlayer}
	\begin{pgfonlayer}{edgelayer}
		\draw [style=rededge] (1) to (5.center);
		\draw [style=rededge, bend right, looseness=0.50] (6) to (7.center);
		\draw [style=rededge] (2.center) to (6);
		\draw [style=rededge, bend left, looseness=0.50] (6) to (3.center);
		\draw [style=rededge, bend left, looseness=0.50] (1) to (10.center);
		\draw [style=rededge, bend right, looseness=0.50] (1) to (8.center);
	\end{pgfonlayer}
\end{tikzpicture}}
\end{center}
We now wish to define monoidal functors $\Lplus\xrightarrow T\Molp\xhookrightarrow{i}\PartMolp$ so as to make this into a system of monoidal categories. First, define a monoidal functor $T:\Lplus\rightarrow\Molp$ by the action on the generating objects in Figure~\ref{fig:translation}, where we use the convention from chemistry that an unlabelled vertex represents a carbon atom with an appropriate number of hydrogen atoms attached to it to make its valence equal to $4$. The only generating morphism of $\Lplus$ is mapped to the only generating morphism of $\Molp$. The monoidal functor $\Molp\xhookrightarrow{i}\PartMolp$ is identity on objects and maps the only generating morphism of $\Molp$ to the composite morphism in the middle rectangle of Figure~\ref{fig:glucose}.
\begin{figure}
\centering
\begin{tabular}{ c c c | c c c }
$\glucose$ & $\mapsto$ & \scalebox{.5}{\begin{tikzpicture}
	\begin{pgfonlayer}{nodelayer}
		\node [style=textbox] (11) at (2.25, -1.75) {\quad OH};
		\node [style=textbox] (10) at (1, -3) {\quad OH};
		\node [style=textbox] (9) at (-1, -0.5) {\quad OH};
		\node [style=textbox] (8) at (-2.25, -1.75) {HO \qquad};
		\node [style=textbox] (1) at (1, 1) {O};
		\node [style=textbox] (7) at (-2.25, 3) {HO};
		\node [style=none] (12) at (-1, 1) {};
		\node [style=none] (13) at (-1, -1.75) {};
		\node [style=none] (14) at (1, -1.75) {};
		\node [style=none] (15) at (-2.25, -0.5) {};
		\node [style=none] (16) at (2.25, -0.5) {};
		\node [style=none] (17) at (-1, 2) {};
	\end{pgfonlayer}
	\begin{pgfonlayer}{edgelayer}
		\draw [style=thedge] (7) to (17.center);
		\draw [style=thedge] (17.center) to (12.center);
		\draw [style=thedge] (12.center) to (1);
		\draw [style=thedge] (1) to (16.center);
		\draw [style=thedge] (16.center) to (14.center);
		\draw [style=thedge] (14.center) to (13.center);
		\draw [style=thedge] (13.center) to (15.center);
		\draw [style=thedge] (15.center) to (12.center);
		\draw [style=thedge] (15.center) to (8);
		\draw [style=thedge] (14.center) to (10);
		\draw [style=thedge] (16.center) to (11);
		\draw [style=thedge] (9) to (13.center);
	\end{pgfonlayer}
\end{tikzpicture}} & $\ATP$ & $\mapsto$ & \scalebox{.5}{\begin{tikzpicture}
	\begin{pgfonlayer}{nodelayer}
		\node [style=textbox] (10) at (3.25, -3.25) {\quad OH};
		\node [style=textbox] (9) at (1.25, -3.25) {\quad OH};
		\node [style=textbox] (1) at (2.25, 0.25) {O};
		\node [style=textbox] (7) at (-0.75, 1.75) {O};
		\node [style=none] (13) at (1.25, -2) {};
		\node [style=none] (14) at (3.25, -2) {};
		\node [style=none] (15) at (0.5, -0.75) {};
		\node [style=none] (16) at (4, -0.75) {};
		\node [style=none] (17) at (0.5, 0.75) {};
		\node [style=textbox] (18) at (-2.25, 1.75) {P};
		\node [style=textbox] (19) at (-2.25, 3.25) {O};
		\node [style=textbox] (21) at (-2.25, 0.25) {\ \supmin O};
		\node [style=textbox] (22) at (-3.75, 1.75) {O};
		\node [style=textbox] (23) at (-5.25, 1.75) {P};
		\node [style=textbox] (24) at (-5.25, 3.25) {O};
		\node [style=textbox] (25) at (-5.25, 0.25) {\ \supmin O};
		\node [style=textbox] (26) at (-6.75, 1.75) {O};
		\node [style=textbox] (27) at (-8.25, 1.75) {P};
		\node [style=textbox] (28) at (-8.25, 3.25) {O};
		\node [style=textbox] (29) at (-8.25, 0.25) {\ \supmin O};
		\node [style=textbox] (30) at (-9.75, 1.75) {\supmin O};
		\node [style=textbox] (31) at (4, 0.75) {N};
		\node [style=none] (32) at (3, 2) {};
		\node [style=textbox] (33) at (4, 3.25) {N};
		\node [style=none] (34) at (5.25, 2.75) {};
		\node [style=none] (35) at (5.25, 1.25) {};
		\node [style=textbox] (36) at (6.75, 0.5) {N};
		\node [style=none] (37) at (8.25, 1.25) {};
		\node [style=textbox] (38) at (8.25, 2.75) {N};
		\node [style=none] (39) at (6.75, 3.5) {};
		\node [style=textbox] (40) at (6.75, 5) {\quad NH\textsubscript 2};
	\end{pgfonlayer}
	\begin{pgfonlayer}{edgelayer}
		\draw [style=thedge] (7) to (17.center);
		\draw [style=thedge] (1) to (16.center);
		\draw [style=thedge] (16.center) to (14.center);
		\draw [style=thedge] (14.center) to (13.center);
		\draw [style=thedge] (13.center) to (15.center);
		\draw [style=thedge] (14.center) to (10);
		\draw [style=thedge] (9) to (13.center);
		\draw [style=thedge] (18) to (21);
		\draw [style=doubedge] (18) to (19);
		\draw [style=thedge] (18) to (7);
		\draw [style=thedge] (15.center) to (1);
		\draw [style=thedge] (17.center) to (15.center);
		\draw [style=thedge] (23) to (25);
		\draw [style=doubedge] (23) to (24);
		\draw [style=thedge] (23) to (22);
		\draw [style=thedge] (27) to (29);
		\draw [style=doubedge] (27) to (28);
		\draw [style=thedge] (27) to (26);
		\draw [style=thedge] (18) to (22);
		\draw [style=thedge] (23) to (26);
		\draw [style=thedge] (27) to (30);
		\draw [style=thedge] (31) to (16.center);
		\draw [style=thedge] (31) to (32.center);
		\draw [style=thedge] (31) to (35.center);
		\draw [style=thedge] (34.center) to (33);
		\draw [style=thedge] (34.center) to (39.center);
		\draw [style=thedge] (39.center) to (40);
		\draw [style=thedge] (38) to (37.center);
		\draw [style=thedge] (36) to (35.center);
		\draw [style=doubedge] (33) to (32.center);
		\draw [style=doubedge] (34.center) to (35.center);
		\draw [style=doubedge] (36) to (37.center);
		\draw [style=doubedge] (38) to (39.center);
	\end{pgfonlayer}
\end{tikzpicture}} \\
$\glucosesix$ & $\mapsto$ & \scalebox{.5}{\begin{tikzpicture}
	\begin{pgfonlayer}{nodelayer}
		\node [style=textbox] (11) at (2.25, -1.75) {\quad OH};
		\node [style=textbox] (10) at (1, -3) {\quad OH};
		\node [style=textbox] (9) at (-1, -0.5) {\quad OH};
		\node [style=textbox] (8) at (-2.25, -1.75) {HO \qquad};
		\node [style=textbox] (1) at (1, 1) {O};
		\node [style=textbox] (7) at (-2.25, 3) {O};
		\node [style=none] (12) at (-1, 1) {};
		\node [style=none] (13) at (-1, -1.75) {};
		\node [style=none] (14) at (1, -1.75) {};
		\node [style=none] (15) at (-2.25, -0.5) {};
		\node [style=none] (16) at (2.25, -0.5) {};
		\node [style=none] (17) at (-1, 2) {};
		\node [style=textbox] (18) at (-2.25, 4.5) {P};
		\node [style=textbox] (19) at (-3.75, 4.5) {O};
		\node [style=textbox] (21) at (-2.25, 6) {\ \supmin O};
		\node [style=textbox] (22) at (-0.75, 4.5) {\supmin O};
	\end{pgfonlayer}
	\begin{pgfonlayer}{edgelayer}
		\draw [style=thedge] (7) to (17.center);
		\draw [style=thedge] (17.center) to (12.center);
		\draw [style=thedge] (12.center) to (1);
		\draw [style=thedge] (1) to (16.center);
		\draw [style=thedge] (16.center) to (14.center);
		\draw [style=thedge] (14.center) to (13.center);
		\draw [style=thedge] (13.center) to (15.center);
		\draw [style=thedge] (15.center) to (12.center);
		\draw [style=thedge] (15.center) to (8);
		\draw [style=thedge] (14.center) to (10);
		\draw [style=thedge] (16.center) to (11);
		\draw [style=thedge] (9) to (13.center);
		\draw [style=thedge] (18) to (21);
		\draw [style=thedge] (18) to (22);
		\draw [style=doubedge] (18) to (19);
		\draw [style=thedge] (18) to (7);
	\end{pgfonlayer}
\end{tikzpicture}} & $\ADP$ & $\mapsto$ & \scalebox{.5}{\begin{tikzpicture}
	\begin{pgfonlayer}{nodelayer}
		\node [style=textbox] (10) at (2, -4) {\quad OH};
		\node [style=textbox] (9) at (0, -4) {\quad OH};
		\node [style=textbox] (1) at (1, -0.5) {O};
		\node [style=textbox] (7) at (-2, 1) {O};
		\node [style=none] (13) at (0, -2.75) {};
		\node [style=none] (14) at (2, -2.75) {};
		\node [style=none] (15) at (-0.75, -1.5) {};
		\node [style=none] (16) at (2.75, -1.5) {};
		\node [style=none] (17) at (-0.75, 0) {};
		\node [style=textbox] (18) at (-3.5, 1) {P};
		\node [style=textbox] (19) at (-3.5, 2.5) {O};
		\node [style=textbox] (21) at (-3.5, -0.5) {\ \supmin O};
		\node [style=textbox] (22) at (-5, 1) {O};
		\node [style=textbox] (23) at (-6.5, 1) {P};
		\node [style=textbox] (24) at (-6.5, 2.5) {O};
		\node [style=textbox] (25) at (-6.5, -0.5) {\ \supmin O};
		\node [style=textbox] (30) at (-8, 1) {\supmin O};
		\node [style=textbox] (31) at (2.75, 0) {N};
		\node [style=none] (32) at (1.75, 1.25) {};
		\node [style=textbox] (33) at (2.75, 2.5) {N};
		\node [style=none] (34) at (4, 2) {};
		\node [style=none] (35) at (4, 0.5) {};
		\node [style=textbox] (36) at (5.5, -0.25) {N};
		\node [style=none] (37) at (7, 0.5) {};
		\node [style=textbox] (38) at (7, 2) {N};
		\node [style=none] (39) at (5.5, 2.75) {};
		\node [style=textbox] (40) at (5.5, 4.25) {\quad NH\textsubscript 2};
	\end{pgfonlayer}
	\begin{pgfonlayer}{edgelayer}
		\draw [style=thedge] (7) to (17.center);
		\draw [style=thedge] (1) to (16.center);
		\draw [style=thedge] (16.center) to (14.center);
		\draw [style=thedge] (14.center) to (13.center);
		\draw [style=thedge] (13.center) to (15.center);
		\draw [style=thedge] (14.center) to (10);
		\draw [style=thedge] (9) to (13.center);
		\draw [style=thedge] (18) to (21);
		\draw [style=doubedge] (18) to (19);
		\draw [style=thedge] (18) to (7);
		\draw [style=thedge] (15.center) to (1);
		\draw [style=thedge] (17.center) to (15.center);
		\draw [style=thedge] (23) to (25);
		\draw [style=doubedge] (23) to (24);
		\draw [style=thedge] (23) to (22);
		\draw [style=thedge] (18) to (22);
		\draw [style=thedge] (31) to (16.center);
		\draw [style=thedge] (31) to (32.center);
		\draw [style=thedge] (31) to (35.center);
		\draw [style=thedge] (34.center) to (33);
		\draw [style=thedge] (34.center) to (39.center);
		\draw [style=thedge] (39.center) to (40);
		\draw [style=thedge] (38) to (37.center);
		\draw [style=thedge] (36) to (35.center);
		\draw [style=doubedge] (33) to (32.center);
		\draw [style=doubedge] (34.center) to (35.center);
		\draw [style=doubedge] (36) to (37.center);
		\draw [style=doubedge] (38) to (39.center);
		\draw [style=thedge] (30) to (23);
	\end{pgfonlayer}
\end{tikzpicture}} \\
$\hydrogenion$ & $\mapsto$ & \scalebox{.7}{\begin{tikzpicture}
	\begin{pgfonlayer}{nodelayer}
		\node [style=textbox] (0) at (0, 0) {\supplus H};
	\end{pgfonlayer}
\end{tikzpicture}} & & &
\end{tabular}
\caption{Translation of English names to chemical graphs.\label{fig:translation}}
\end{figure}
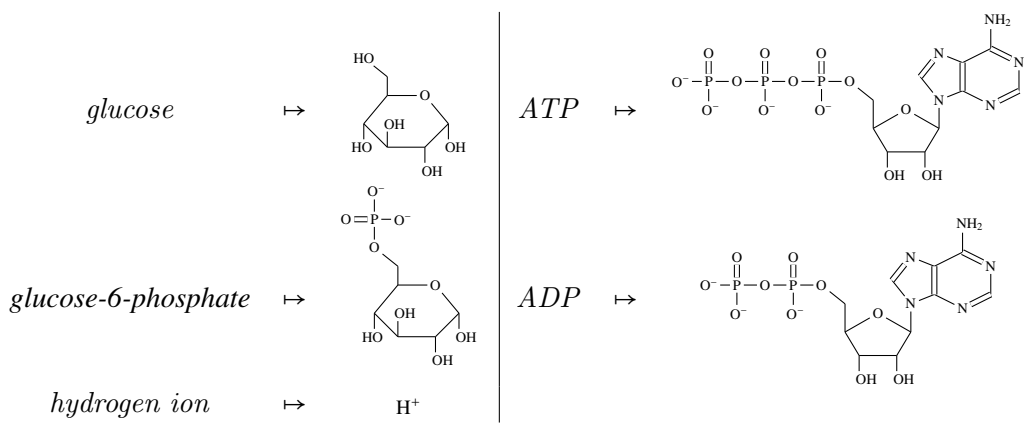

\subsection{Explaining Phosphorylation}
We can now use the lowest level language $\PartMolp$ to explain the high-level rule~\eqref{eqn:phosphorylation} as is shown in Figure~\ref{fig:glucose}. Note that this is indeed an explanation according to Definition~\ref{def:explanation-1}, since the rule that is being explained is internal to $\Lplus$, the explanation does not use any non-identity morphisms from $\Lplus$, and the explanation can be derived starting from the rule~\eqref{eqn:phosphorylation} using the $2$-cells of the layered prop, whose composite gives a $2$-cell from the rule to the explanation.
\begin{figure}[h]
    \centering
    \scalebox{0.7}{
        \input{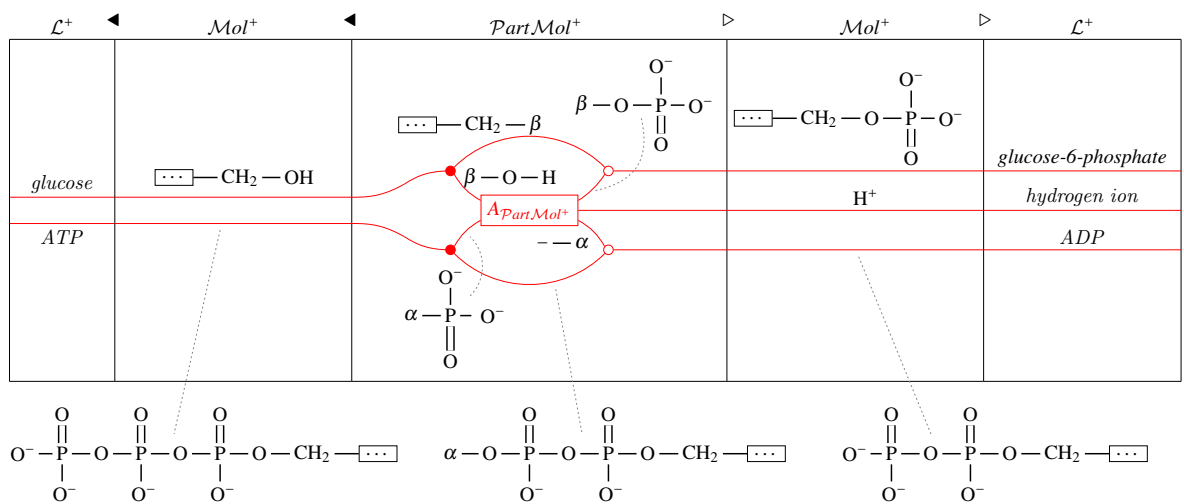}
    }
    \caption{Explaining glucose phosphorylation: each area between the vertical black bars represents a layer, so in this case $\Lplus$, $\Molp$ or $\PartMolp$.\label{fig:glucose}}
\end{figure}
While the diagram in Figure~\ref{fig:glucose} fulfills the definition of an explanation, it is not very ``explanatory" in an intuitive sense. This is because we chose to stop at a fairly high level of abstraction. It is important to note that the morphism $A_{\PartMolp}$ is just a black box, which could itself be explained at the level of atoms exchanging electrons. Modularity of layered props would then allow us to add this further level to the diagram. The resulting explanation would bring us closer to satisfactorily answering the question {\em Why does this reaction occur?}.

We conclude this example by remarking that we didn't have to assume that we already know the higher level chemical rule~\eqref{eqn:phosphorylation}. Instead we could have chosen to {\em generate} the higher level rules by declaring as morphisms every $1$-cell from $(\Lplus,c)$ to $(\Lplus,d)$ for some objects $c,d\in\Lplus$. Instead of an explanation, this would correspond to deriving higher level rules from a single lower level rule.

\section{Example: Electrical Circuits}\label{sec:electrical-circuits}

While in the previous section we constructed a minimal example from scratch, in this section we take an existing example from the literature where explanations are already used implicitly. Namely, we focus on the research program that has formalised electrical circuits in terms of string diagrams and given them an interpretation in graphical affine algebra~\cite{compositional-networks,graphical-affine-algebra,survey-signal-flow,electrical-circuits}.

The string diagrammatic electrical circuit theory is a paradigmatic example of explanations taking a functorial form: the relations between electrical components are proved by interpreting them as morphisms in the graphical affine algebra. Thus this example also shows how functorial explanations can be incorporated into our framework. Note, however, that Boisseau and Soboci\' nski~\cite{electrical-circuits} already use something like layered explanations to only partially translate their diagrams. They call the notational device used for this an {\em impedance box}. In our language, impedance boxes arise in a principled way as instances of a general definition: they are just windows (Definition~\ref{def:window}) of a particular shape.

We define the props {\em graphical affine algebra} $\GAA$ and {\em electrical circuits} $\ECirc$ as well as the translation functor $\mathcal I:\ECirc\rightarrow\GAA$ as in~\cite{electrical-circuits}, except that we quotient the morphisms in $\ECirc$ by equality under $\mathcal I$. This makes $\mathcal I$ faithful, which we reflect in our syntax by adding a left inverse to the $2$-cell \texttt{unit-ref} in Figure~\ref{fig:axioms-pants2}\footnote{This causes some problems for the semantic interpretation of Section~\ref{sec:semantics}, whose resolution we leave for a more technical paper.}. Additionally, we define the {\em impedance category} $\Imp$ and the category of bipoles $\Bip$ in order to express the impedance calculus of~\cite{electrical-circuits} formally within our setup.
\begin{definition}[Impedance category]
  The {\em impedance category} $\Imp$ is a prop whose generating morphisms are all the morphisms of $\GAA$ with exactly one input and exactly one output. The identity is \scalebox{.6}{\begin{tikzpicture}
	\begin{pgfonlayer}{nodelayer}
		\node [style=none] (9) at (-2.25, 0) {};
		\node [style=none] (10) at (2.25, 0) {};
		\node [style=node] (11) at (-0.75, 0) {};
		\node [style=wnode] (12) at (0.75, 0) {};
	\end{pgfonlayer}
	\begin{pgfonlayer}{edgelayer}
		\draw [style=gray edge] (9.center) to (11);
		\draw [style=gray edge] (12) to (10.center);
	\end{pgfonlayer}
\end{tikzpicture}}, and composition is given by the rule
  \begin{center}
  \scalebox{.6}{\begin{tikzpicture}
	\begin{pgfonlayer}{nodelayer}
		\node [style=none] (0) at (-10.5, 0) {};
		\node [style=none] (1) at (-6.5, 0) {};
		\node [style=none] (2) at (-5.75, 0) {$;$};
		\node [style=bdytextbox] (4) at (-8.5, 0) {$C$};
		\node [style=none] (5) at (-5, 0) {};
		\node [style=none] (6) at (-1, 0) {};
		\node [style=bdytextbox] (7) at (-3, 0) {$D$};
		\node [style=none] (8) at (0, 0) {$\coloneqq$};
		\node [style=none] (9) at (1, 0) {};
		\node [style=none] (10) at (8, 0) {};
		\node [style=node] (11) at (2.5, 0) {};
		\node [style=wnode] (12) at (6.5, 0) {};
		\node [style=bdytextbox] (13) at (4.5, 1) {$C$};
		\node [style=bdytextbox] (14) at (4.5, -1) {$D$};
	\end{pgfonlayer}
	\begin{pgfonlayer}{edgelayer}
		\draw [style=gray edge] (0.center) to (4);
		\draw [style=gray edge] (4) to (1.center);
		\draw [style=gray edge] (5.center) to (7);
		\draw [style=gray edge] (7) to (6.center);
		\draw [style=gray edge, bend left=15, looseness=0.75] (11) to (13);
		\draw [style=gray edge, bend left=15, looseness=0.75] (13) to (12);
		\draw [style=gray edge, bend left=15, looseness=0.75] (12) to (14);
		\draw [style=gray edge, bend left=15, looseness=0.75] (14) to (11);
		\draw [style=gray edge] (9.center) to (11);
		\draw [style=gray edge] (12) to (10.center);
	\end{pgfonlayer}
\end{tikzpicture}}.
  \end{center}
\end{definition}
\begin{definition}[Bipole category]
  The {\em bipole category} $\Bip$ is the subcategory of $\ECirc$ given by those generators which have exactly one input and one output. That is, it is the free prop generated by
  \begin{center}
    \begin{tabular}{ c | c | c | c | c }
  \scalebox{0.6}{\begin{tikzpicture}
	\begin{pgfonlayer}{nodelayer}
		\node [style=none] (0) at (-1.25, 0) {};
		\node [style=none] (1) at (1.25, 0) {};
		\node [style=none] (2) at (-1.75, 0) {};
		\node [style=none] (3) at (1.75, 0) {};
		\node [style=bluetext] (4) at (0, 1.25) {$R$};
	\end{pgfonlayer}
	\begin{pgfonlayer}{edgelayer}
		\draw [style=resistor] (0.center) to (1.center);
		\draw [style=blthedge] (2.center) to (0.center);
		\draw [style=blthedge] (1.center) to (3.center);
	\end{pgfonlayer}
\end{tikzpicture}} & \scalebox{0.6}{\begin{tikzpicture}
	\begin{pgfonlayer}{nodelayer}
		\node [style=none] (0) at (-1.25, 0) {};
		\node [style=none] (1) at (1.25, 0) {};
		\node [style=none] (2) at (-1.75, 0) {};
		\node [style=none] (3) at (1.75, 0) {};
		\node [style=bluetext] (4) at (0, 1.25) {$L$};
	\end{pgfonlayer}
	\begin{pgfonlayer}{edgelayer}
		\draw [style=inductor] (0.center) to (1.center);
		\draw [style=blthedge] (2.center) to (0.center);
		\draw [style=blthedge] (1.center) to (3.center);
	\end{pgfonlayer}
\end{tikzpicture}} & \scalebox{0.6}{\begin{tikzpicture}
	\begin{pgfonlayer}{nodelayer}
		\node [style=none] (0) at (-0.25, 0) {};
		\node [style=none] (1) at (0.25, 0) {};
		\node [style=none] (2) at (-1.5, 0) {};
		\node [style=none] (3) at (1.5, 0) {};
		\node [style=bluetext] (4) at (0, 1.5) {$C$};
	\end{pgfonlayer}
	\begin{pgfonlayer}{edgelayer}
		\draw [style=capacitor] (0.center) to (1.center);
		\draw [style=blthedge] (2.center) to (0.center);
		\draw [style=blthedge] (1.center) to (3.center);
	\end{pgfonlayer}
\end{tikzpicture}} & \scalebox{0.6}{\begin{tikzpicture}
	\begin{pgfonlayer}{nodelayer}
		\node [style=none] (0) at (-0.75, 0) {};
		\node [style=none] (1) at (0.75, 0) {};
		\node [style=none] (2) at (-1.75, 0) {};
		\node [style=none] (3) at (1.75, 0) {};
		\node [style=bluetext] (4) at (0, 1.5) {$V$};
	\end{pgfonlayer}
	\begin{pgfonlayer}{edgelayer}
		\draw [style=voltage-source] (0.center) to (1.center);
		\draw [style=blthedge] (2.center) to (0.center);
		\draw [style=blthedge] (1.center) to (3.center);
	\end{pgfonlayer}
\end{tikzpicture}} & \scalebox{0.6}{\begin{tikzpicture}
	\begin{pgfonlayer}{nodelayer}
		\node [style=none] (0) at (-0.75, 0) {};
		\node [style=none] (1) at (0.75, 0) {};
		\node [style=none] (2) at (-1.75, 0) {};
		\node [style=none] (3) at (1.75, 0) {};
		\node [style=bluetext] (4) at (0, 1.5) {$I$};
	\end{pgfonlayer}
	\begin{pgfonlayer}{edgelayer}
		\draw [style=current-source] (0.center) to (1.center);
		\draw [style=blthedge] (2.center) to (0.center);
		\draw [style=blthedge] (1.center) to (3.center);
	\end{pgfonlayer}
\end{tikzpicture}}.
\end{tabular}
\end{center}
\end{definition}

Define the ``boxing" functor $B:\Bip\rightarrow\Imp$ by the following action on the generators:
\begin{center}
\scalebox{.6}{\begin{tikzpicture}
	\begin{pgfonlayer}{nodelayer}
		\node [style=none] (0) at (-18, 4) {};
		\node [style=none] (1) at (-15.5, 4) {};
		\node [style=none] (2) at (-18.5, 4) {};
		\node [style=none] (3) at (-15, 4) {};
		\node [style=bluetext] (4) at (-16.75, 5.25) {$R$};
		\node [style=none] (5) at (-13, 4) {\scalebox{1.3}{$\mapsto$}};
		\node [style=coreg] (6) at (-9.25, 4) {$R$};
		\node [style=none] (7) at (-10.75, 4) {};
		\node [style=none] (8) at (-7.75, 4) {};
		\node [style=none] (9) at (-5, 4) {};
		\node [style=none] (10) at (-2.5, 4) {};
		\node [style=none] (11) at (-5.5, 4) {};
		\node [style=none] (12) at (-2, 4) {};
		\node [style=bluetext] (13) at (-3.75, 5.25) {$L$};
		\node [style=coreg] (15) at (3.75, 4) {$Lx$};
		\node [style=none] (16) at (2.25, 4) {};
		\node [style=none] (17) at (5.25, 4) {};
		\node [style=none] (18) at (8.5, 4) {};
		\node [style=none] (19) at (9, 4) {};
		\node [style=none] (20) at (7.25, 4) {};
		\node [style=none] (21) at (10.25, 4) {};
		\node [style=bluetext] (22) at (8.75, 5.5) {$C$};
		\node [style=none] (25) at (14.75, 4) {};
		\node [style=none] (26) at (17.75, 4) {};
		\node [style=reg] (27) at (16.25, 4) {$Cx$};
		\node [style=none] (28) at (-11.75, 0) {};
		\node [style=none] (29) at (-10.25, 0) {};
		\node [style=none] (30) at (-12.75, 0) {};
		\node [style=none] (31) at (-9.25, 0) {};
		\node [style=bluetext] (32) at (-11, 1.5) {$V$};
		\node [style=coreg] (34) at (-2.25, 0) {$V$};
		\node [style=none] (35) at (-5.25, 0) {};
		\node [style=none] (36) at (-1.25, 0) {};
		\node [style=node] (37) at (-4.25, 0) {};
		\node [style=none] (38) at (-3.25, 0) {};
		\node [style=none] (39) at (-3.25, 0.25) {};
		\node [style=none] (40) at (-3.25, -0.25) {};
		\node [style=none] (41) at (2.25, 0) {};
		\node [style=none] (42) at (3.75, 0) {};
		\node [style=none] (43) at (1.25, 0) {};
		\node [style=none] (44) at (4.75, 0) {};
		\node [style=bluetext] (45) at (3, 1.5) {$I$};
		\node [style=none] (48) at (12.5, 0) {};
		\node [style=none] (49) at (8.5, 0) {};
		\node [style=node] (50) at (11.5, 0) {};
		\node [style=none] (51) at (10.5, 0) {};
		\node [style=none] (52) at (10.5, 0.25) {};
		\node [style=none] (53) at (10.5, -0.25) {};
		\node [style=reg] (54) at (9.5, 0) {$I$};
		\node [style=none] (55) at (13, 0) {{\large . }};
		\node [style=none] (56) at (-7.25, 4) {{\large ,}};
		\node [style=none] (57) at (5.75, 4) {{\large ,}};
		\node [style=none] (58) at (18.25, 4) {{\large ,}};
		\node [style=none] (59) at (-0.75, 0) {{\large ,}};
		\node [style=none] (60) at (0, 4) {\scalebox{1.3}{$\mapsto$}};
		\node [style=none] (61) at (12.5, 4) {\scalebox{1.3}{$\mapsto$}};
		\node [style=none] (62) at (-7.25, 0) {\scalebox{1.3}{$\mapsto$}};
		\node [style=none] (63) at (6.5, 0) {\scalebox{1.3}{$\mapsto$}};
	\end{pgfonlayer}
	\begin{pgfonlayer}{edgelayer}
		\draw [style=resistor] (0.center) to (1.center);
		\draw [style=blthedge] (2.center) to (0.center);
		\draw [style=blthedge] (1.center) to (3.center);
		\draw [style=gray edge] (7.center) to (6);
		\draw [style=gray edge] (6) to (8.center);
		\draw [style=inductor] (9.center) to (10.center);
		\draw [style=blthedge] (11.center) to (9.center);
		\draw [style=blthedge] (10.center) to (12.center);
		\draw [style=gray edge] (16.center) to (15);
		\draw [style=gray edge] (15) to (17.center);
		\draw [style=capacitor] (18.center) to (19.center);
		\draw [style=blthedge] (20.center) to (18.center);
		\draw [style=blthedge] (19.center) to (21.center);
		\draw [style=gray edge] (25.center) to (27);
		\draw [style=gray edge] (27) to (26.center);
		\draw [style=voltage-source] (28.center) to (29.center);
		\draw [style=blthedge] (30.center) to (28.center);
		\draw [style=blthedge] (29.center) to (31.center);
		\draw [style=gray edge] (34) to (36.center);
		\draw [style=gray edge] (35.center) to (37);
		\draw [style=gray edge] (34) to (38.center);
		\draw [style=edge] (39.center) to (40.center);
		\draw [style=current-source] (41.center) to (42.center);
		\draw [style=blthedge] (43.center) to (41.center);
		\draw [style=blthedge] (42.center) to (44.center);
		\draw [style=gray edge] (48.center) to (50);
		\draw [style=edge] (52.center) to (53.center);
		\draw [style=gray edge] (49.center) to (54);
		\draw [style=gray edge] (54) to (51.center);
	\end{pgfonlayer}
\end{tikzpicture}}
\end{center}
Further, define a "wrapping" functor $W:\Imp\rightarrow\GAA$ which acts as $n\mapsto 2n$ on objects and on morphisms as shown below left. The boxing and the wrapping functors are so defined that we have a commutative square below right:
\begin{center}
\scalebox{.6}{\begin{tikzpicture}
	\begin{pgfonlayer}{nodelayer}
		\node [style=bdytextbox] (0) at (-4.5, 0) {$C$};
		\node [style=none] (1) at (-6.5, 0) {};
		\node [style=none] (2) at (-2.5, 0) {};
		\node [style=none] (3) at (0, 0) {$\mapsto$};
		\node [style=none] (4) at (2.5, 1.75) {};
		\node [style=none] (5) at (2.5, -0.75) {};
		\node [style=none] (6) at (8.5, 0.75) {};
		\node [style=none] (7) at (8.5, -1.75) {};
		\node [style=bdytextbox] (8) at (5.5, 0) {$C$};
		\node [style=node] (9) at (4, -0.75) {};
		\node [style=wnode] (10) at (7, 0.75) {};
	\end{pgfonlayer}
	\begin{pgfonlayer}{edgelayer}
		\draw [style=gray edge] (1.center) to (0);
		\draw [style=gray edge] (0) to (2.center);
		\draw [style=edge] (5.center) to (9);
		\draw [style=edge, bend left=15] (9) to (8);
		\draw [style=edge, bend right=15, looseness=0.75] (9) to (7.center);
		\draw [style=edge] (6.center) to (10);
		\draw [style=edge, bend left=15] (10) to (8);
		\draw [style=edge, bend right=15, looseness=0.75] (10) to (4.center);
	\end{pgfonlayer}
\end{tikzpicture}}\qquad\qquad\qquad\scalebox{.8}{\begin{tikzpicture}
	\begin{pgfonlayer}{nodelayer}
		\node [style=none] (4) at (-2.75, 0) {$B$};
		\node [style=none] (6) at (2.75, 0) {$\mathcal I$};
		\node [style=none] (7) at (0, -2.25) {$W$};
		\node [style=objectbox] (8) at (2, -1.5) {$\GAA$};
		\node [style=objectbox] (9) at (-2, -1.5) {$\Imp$};
		\node [style=objectbox] (10) at (-2, 1.5) {$\Bip$};
		\node [style=objectbox] (11) at (2, 1.5) {$\ECirc$};
	\end{pgfonlayer}
	\begin{pgfonlayer}{edgelayer}
		\draw [style=morphism] (9) to (8);
		\draw [style=morphism] (10) to (9);
		\draw [style=morphism] (11) to (8);
		\draw [style=hookarrow] (10) to (11);
	\end{pgfonlayer}
\end{tikzpicture}}
\end{center}
where the top horizontal morphism is the inclusion functor, and $\mathcal I$ is the translation of electrical circuits to graphical affine algebra. Treating the above diagram of monoidal functors as a system of monoidal categories, we obtain a layered prop. Within this layered prop, we are able to replicate what is called the {\em impedance calculus} in~\cite{electrical-circuits}. To illustrate this, we give an explanation of the rule governing the sequential composition of resistors. This rule is a $2$-cell in the layered prop, and the explanation is therefore that of a $2$-cell (Definition~\ref{def:explanation-2}).

Figure~\ref{fig:resistors} shows how the rule for composing two resistors
\begin{center}
\scalebox{0.6}{\begin{tikzpicture}
	\begin{pgfonlayer}{nodelayer}
		\node [style=none] (0) at (-9.75, 0) {};
		\node [style=none] (1) at (-7.25, 0) {};
		\node [style=none] (2) at (-10.75, 0) {};
		\node [style=bluetext] (4) at (-8.5, 1) {$R1$};
		\node [style=none] (5) at (-6.25, 0) {};
		\node [style=bluetext] (7) at (-5, 1) {$R2$};
		\node [style=none] (8) at (-2.5, 0) {};
		\node [style=none] (9) at (-3.75, 0) {};
		\node [style=none] (15) at (-1, 0.25) {};
		\node [style=none] (16) at (-1, -0.25) {};
		\node [style=none] (17) at (1, 0.25) {};
		\node [style=none] (18) at (1, -0.25) {};
		\node [style=none] (126) at (4.5, 0) {};
		\node [style=none] (128) at (2.5, 0) {};
		\node [style=bluetext] (129) at (5.75, 1) {$R1 + R2$};
		\node [style=none] (132) at (9, 0) {};
		\node [style=none] (133) at (7, 0) {};
	\end{pgfonlayer}
	\begin{pgfonlayer}{edgelayer}
		\draw [style=resistor] (0.center) to (1.center);
		\draw [style=blthedge] (0.center) to (2.center);
		\draw [style=blthedge] (1.center) to (5.center);
		\draw [style=blthedge] (8.center) to (9.center);
		\draw [style=resistor] (5.center) to (9.center);
		\draw [style=diredge] (15.center) to (17.center);
		\draw [style=diredge] (18.center) to (16.center);
		\draw [style=blthedge] (126.center) to (128.center);
		\draw [style=blthedge] (132.center) to (133.center);
		\draw [style=resistor] (126.center) to (133.center);
	\end{pgfonlayer}
\end{tikzpicture}}
\end{center}
can be explained (this is essentially part~(i) of Proposition~3 of~\cite{electrical-circuits}). This is indeed an explanation of a $2$-cell (Definition~\ref{def:explanation-2}), since we are explaining an equality in $\ECirc$ using only the $2$-cells of a layered prop and a $2$-cell from $\Imp$ (the third $2$-cell of the derivation).
\begin{figure}[h]
    \centering
    \scalebox{0.7}{
        \input{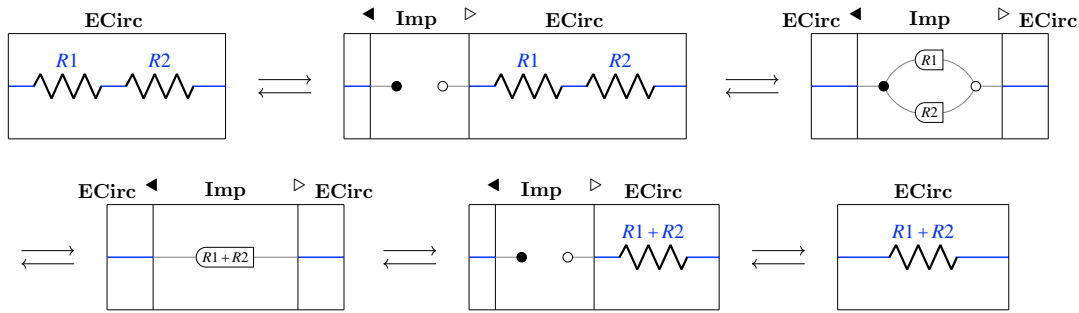}
    }
    \caption{Explaining sequential composition of resistors. Note that the explanation relies on the composition in $\Imp$. This could, in turn, be itself explained by translating to $\GAA$.\label{fig:resistors}}
\end{figure}

As for the example with glucose phosphorylation, we could choose to generate the equalities in $\ECirc$ rather than assume them a priori. In this case, there would be no need to quotient morphisms in $\ECirc$ by equality under the translation functor, yet the equality of $1$-cells should be taken up to a trivial window.

\section{Example: Calculus of Communicating Systems}\label{sec:concurrency}

The calculus of communicating systems (CCS)~\cite{milner} is widely used to reason about programs, formal languages and concurrency. Here we consider a restricted version of CCS and two ways to give semantics to the CCS expressions: {\em reduction semantics} is very heavily syntactic, in addition to the structural congruences, it only allows for only one rewrite rule (the {\em reduction}), while the {\em labelled transition system} (LTS) semantics~\cite{milner} is more flexible and comes with more rewrite rules. Our goal is to show how the LTS semantics can be used to give an explanation (this time in the sense of Definition~\ref{def:explanation-1}) of the rewrite rule of the reduction semantics. Intuitively, the LTS semantics may be seen as a lower level implementation of the concurrent processes described abstractly by CCS. Furthermore, we demonstrate that LTS semantics has a larger scope of allowed derivations than the reduction semantics by giving a counterfactual explanation of a rewrite rule in the reduction semantics.

Let us fix a set of {\em action names} $A$. Define $\bar A\coloneqq\{\bar a : a\in A\}$ and $Act\coloneqq A\cup\bar A\cup\{\tau\}$. The set of {\em processes} is defined recursively as follows, where $x$ ranges over $Act$:
\begin{center}
\begin{tabular}{ c c | c | c}
  $P\Coloneqq$ & $0$ & $x.P$ & $P\parallel P$.
\end{tabular}
\end{center}

\begin{definition}[Congruence]
  Define the {\em congruence} as the smallest equivalence relation $\sim$ on the set of processes that satisfies:
  \begin{center}
  \begin{tabular}{ c c c }
    $P\parallel Q\sim Q\parallel P$, & \quad & $(P\parallel Q)\parallel R\sim P\parallel (Q\parallel R)$, \\ \\
    $0\parallel P\sim P$, & \quad & if $P\sim P'$ and $Q\sim Q'$, then $P\parallel Q\sim P'\parallel Q'$.
  \end{tabular}
\end{center}
\end{definition}

\begin{definition}[Reduction semantics]
  A {\em rewrite rule} in {\em reduction semantics} is an ordered pair of processes, which we write as $P\rightarrow Q$, generated by the following three {\em deduction rules}:
  \begin{center}
  \begin{tabular}{ c c c }
    \begin{tabular}{ c }
      \\
      \hline
      $x.P\parallel\bar x.Q\rightarrow P\parallel Q$
    \end{tabular}
    &
    \begin{tabular}{ c }
      $P\rightarrow Q$ \\
      \hline
      $P\parallel R\rightarrow Q\parallel R$
    \end{tabular}
    &
    \begin{tabular}{ c }
      $P\rightarrow Q$\quad $P\sim P'$\quad $Q\sim Q'$ \\
      \hline
      $P'\rightarrow Q'$
    \end{tabular}
  \end{tabular}
\end{center}
\end{definition}

In other words, rewrite rules are parallel compositions of the {\em reduction} (first rule in the above definition) up to the congruence. For instance, we can derive the following rewrite rule:
\begin{equation}\label{eqn:CCS-derivation}
x.0\parallel (y.0\parallel\bar x.0)\rightarrow 0\parallel (y.0\parallel 0).
\end{equation}

In order to talk about layered props, we wish to express reduction semantics as a monoidal category. Let $\Red$ be the monoidal category whose objects are the processes, monoidal product on objects is the parallel composition $\parallel$, and whose morphisms are generated by:
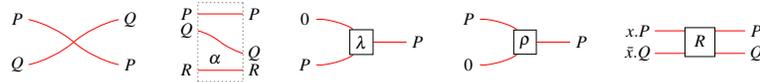
\begin{center}
\scalebox{.6}{\begin{tikzpicture}
	\begin{pgfonlayer}{nodelayer}
		\node [style=none] (29) at (-15, 1) {};
		\node [style=none] (30) at (-15, -1) {};
		\node [style=none] (31) at (-13, 0) {};
		\node [style=none] (32) at (-11, 1) {};
		\node [style=none] (33) at (-11, -1) {};
		\node [style=none] (34) at (-15.5, 1) {$P$};
		\node [style=none] (35) at (-15.5, -1) {$Q$};
		\node [style=none] (36) at (-10.5, 1) {$Q$};
		\node [style=none] (37) at (-10.5, -1) {$P$};
		\node [style=none] (38) at (12.75, 0.5) {};
		\node [style=none] (39) at (12.75, -0.5) {};
		\node [style=none] (40) at (16.75, 0.5) {};
		\node [style=none] (41) at (16.75, -0.5) {};
		\node [style=bdytextbox] (42) at (14.75, 0) {$\phantom{|}R\phantom{|}$};
		\node [style=none] (43) at (12, 0.5) {$x.P$};
		\node [style=none] (44) at (12, -0.5) {$\bar x.Q$};
		\node [style=none] (45) at (17.25, 0.5) {$P$};
		\node [style=none] (46) at (17.25, -0.5) {$Q$};
		\node [style=bdytextbox] (47) at (-0.25, 0) {$\lambda$};
		\node [style=none] (48) at (-2.25, 1) {};
		\node [style=none] (49) at (-2.25, -1) {};
		\node [style=none] (50) at (1.75, 0) {};
		\node [style=none] (51) at (-2.75, 1) {$0$};
		\node [style=none] (52) at (-2.75, -1) {$P$};
		\node [style=none] (53) at (2.25, 0) {$P$};
		\node [style=bdytextbox] (54) at (7, 0) {$\rho$};
		\node [style=none] (55) at (5, -1) {};
		\node [style=none] (56) at (5, 1) {};
		\node [style=none] (57) at (9, 0) {};
		\node [style=none] (58) at (4.5, -1) {$0$};
		\node [style=none] (59) at (4.5, 1) {$P$};
		\node [style=none] (60) at (9.5, 0) {$P$};
		\node [style=none] (67) at (-8, -1.25) {$R$};
		\node [style=none] (68) at (-5, 1.25) {$P$};
		\node [style=none] (70) at (-7.5, 0.5) {};
		\node [style=none] (71) at (-6.5, 0) {};
		\node [style=none] (72) at (-5.5, -0.5) {};
		\node [style=none] (73) at (-5.5, -1.25) {};
		\node [style=none] (74) at (-7.5, 1.75) {};
		\node [style=none] (75) at (-5.5, 1.75) {};
		\node [style=none] (76) at (-7.5, -1.75) {};
		\node [style=none] (77) at (-5.5, -1.75) {};
		\node [style=none] (78) at (-6.75, -0.75) {$\alpha$};
		\node [style=none] (79) at (-7.5, 1.25) {};
		\node [style=none] (80) at (-5.5, 1.25) {};
		\node [style=none] (81) at (-7.5, -1.25) {};
		\node [style=none] (82) at (-5, -1.25) {$R$};
		\node [style=none] (83) at (-8, 1.25) {$P$};
		\node [style=none] (84) at (-8, 0.5) {$Q$};
		\node [style=none] (85) at (-5, -0.5) {$Q$};
	\end{pgfonlayer}
	\begin{pgfonlayer}{edgelayer}
		\draw [style=rededge, bend left=15, looseness=0.75] (29.center) to (31.center);
		\draw [style=rededge, bend left=15, looseness=0.75] (31.center) to (30.center);
		\draw [style=rededge, bend left=15, looseness=0.75] (31.center) to (32.center);
		\draw [style=rededge, bend right=15, looseness=0.75] (31.center) to (33.center);
		\draw [style=rededge] (38.center) to (40.center);
		\draw [style=rededge] (39.center) to (41.center);
		\draw [style=rededge, bend left, looseness=0.50] (48.center) to (47);
		\draw [style=rededge, bend left, looseness=0.50] (47) to (49.center);
		\draw [style=rededge] (47) to (50.center);
		\draw [style=rededge, bend right, looseness=0.50] (55.center) to (54);
		\draw [style=rededge, bend right, looseness=0.50] (54) to (56.center);
		\draw [style=rededge] (54) to (57.center);
		\draw [style=rededge, bend left=15, looseness=0.75] (70.center) to (71.center);
		\draw [style=rededge, bend right=15, looseness=0.50] (71.center) to (72.center);
		\draw [style=gray dashed] (74.center) to (75.center);
		\draw [style=gray dashed] (75.center) to (77.center);
		\draw [style=gray dashed] (77.center) to (76.center);
		\draw [style=gray dashed] (76.center) to (74.center);
		\draw [style=rededge] (79.center) to (80.center);
		\draw [style=rededge] (81.center) to (73.center);
	\end{pgfonlayer}
\end{tikzpicture}}
\end{center}
together with inverses for the first four generators. Here $P$, $Q$ and $R$ range over processes, and $x$ ranges over $A$. The first four morphisms correspond to the congruence, and $R$ corresponds to the first deduction rule for transitions. The parallel composition is taken care of by the monoidal structure. Note that the monoidal product is not strictly associative, so we need to keep track of the bracketing of the wires.

Next, we introduce a LTS as an alternative semantics for the above fragment of CCS.
\begin{definition}[Labelled transition]
  A {\em labelled transition} is a triple $(P,x,Q)$, where $P$ and $Q$ are processes and $x\in Act$, generated by the deduction rules below. We write $P\xrightarrow x Q$ for such triple. Note that we write the silent action $\tau$ as an unlabelled arrow.
  \begin{center}
  \begin{tabular}{ c c c c }
    \begin{tabular}{ c }
      $P'\xrightarrow x P$ \\
      \hline
      $P'\parallel Q\xrightarrow x P\parallel Q$
    \end{tabular}
    &
    \begin{tabular}{ c }
      $P'\xrightarrow x P$ \\
      \hline
      $Q\parallel P'\xrightarrow x Q\parallel P$
    \end{tabular}
    &
    \begin{tabular}{ c }
      \\
      \hline
      $x.P\xrightarrow x P$
    \end{tabular}
    &
    \begin{tabular}{ c }
      $P'\xrightarrow x P$\quad $Q'\xrightarrow{\bar x}Q$ \\
      \hline
      $P'\parallel Q'\rightarrow P\parallel Q$
    \end{tabular}
  \end{tabular}
\end{center}
\end{definition}
\begin{definition}[Bisimulation]
  A {\em bisimulation} on the set of processes is a binary relation $b$ such that for all processes $P$ and $Q$ and all $x\in Act$, we have that $PbQ$ implies
  \begin{itemize}[topsep=0pt,itemsep=-1ex,partopsep=1ex,parsep=1ex]
    \item if $P\xrightarrow x P'$, then there is a process $Q'$ such that $Q\xrightarrow x Q'$ and $P'bQ'$,
    \item if $Q\xrightarrow x Q'$, then there is a process $P'$ such that $P\xrightarrow x P'$ and $P'bQ'$.
  \end{itemize}
  The {\em largest bisimulation} is the union of all bisimulations.
\end{definition}

Labelled transitions define the {\em LTS semantics}, which, similarly to the reduction semantics, can be modelled as a monoidal category. Thus let $\LTS$ be the free monoidal category whose generating objects are pairs $P,\uparrow x$, where $P$ is a process and $x\in Act$. We think of $\uparrow x$ as the ``pending action", and omit the silent pending action: $P\coloneqq P,\uparrow{\tau}$. The morphisms of $\LTS$ are generated by
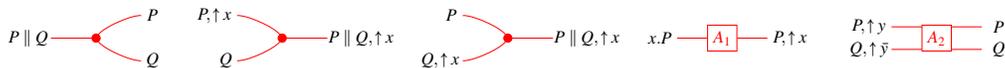
\begin{center}
\scalebox{.6}{\begin{tikzpicture}
	\begin{pgfonlayer}{nodelayer}
		\node [style=redbox] (0) at (7.25, 0) {$A_1$};
		\node [style=none] (1) at (5.25, 0) {};
		\node [style=none] (2) at (9.25, 0) {};
		\node [style=redbox] (3) at (16.75, 0) {$A_2$};
		\node [style=none] (4) at (14.75, 0.5) {};
		\node [style=none] (5) at (14.75, -0.5) {};
		\node [style=none] (6) at (18.75, 0.5) {};
		\node [style=none] (7) at (18.75, -0.5) {};
		\node [style=none] (8) at (4.5, 0) {$x.P$};
		\node [style=none] (9) at (10.25, 0) {$P,\uparrow x$};
		\node [style=none] (10) at (13.75, 0.5) {$P,\uparrow y$};
		\node [style=none] (11) at (13.75, -0.5) {$Q,\uparrow{\bar y}$};
		\node [style=none] (12) at (19.5, 0.5) {$P$};
		\node [style=none] (13) at (19.5, -0.5) {$Q$};
		\node [style=rednode] (14) at (-20.5, 0) {};
		\node [style=none] (15) at (-22.5, 0) {};
		\node [style=none] (16) at (-18.5, 1) {};
		\node [style=none] (17) at (-18.5, -1) {};
		\node [style=none] (18) at (-23.5, 0) {$P\parallel Q$};
		\node [style=none] (19) at (-18, 1) {$P$};
		\node [style=none] (20) at (-18, -1) {$Q$};
		\node [style=rednode] (21) at (-12.25, 0) {};
		\node [style=none] (22) at (-10.25, 0) {};
		\node [style=none] (23) at (-14.25, 1) {};
		\node [style=none] (24) at (-14.25, -1) {};
		\node [style=none] (25) at (-8.75, 0) {$P\parallel Q,\uparrow x$};
		\node [style=none] (26) at (-15.25, 1) {$P,\uparrow x$};
		\node [style=none] (27) at (-14.75, -1) {$Q$};
		\node [style=rednode] (28) at (-2.25, 0) {};
		\node [style=none] (29) at (-0.25, 0) {};
		\node [style=none] (30) at (-4.25, 1) {};
		\node [style=none] (31) at (-4.25, -1) {};
		\node [style=none] (32) at (1.25, 0) {$P\parallel Q,\uparrow x$};
		\node [style=none] (33) at (-4.75, 1) {$P$};
		\node [style=none] (34) at (-5.25, -1) {$Q,\uparrow x$};
	\end{pgfonlayer}
	\begin{pgfonlayer}{edgelayer}
		\draw [style=rededge] (1.center) to (0);
		\draw [style=rededge] (0) to (2.center);
		\draw [style=rededge] (4.center) to (6.center);
		\draw [style=rededge] (5.center) to (7.center);
		\draw [style=rededge] (15.center) to (14);
		\draw [style=rededge, bend left, looseness=0.50] (14) to (16.center);
		\draw [style=rededge, bend right, looseness=0.50] (14) to (17.center);
		\draw [style=rededge] (22.center) to (21);
		\draw [style=rededge, bend right, looseness=0.50] (21) to (23.center);
		\draw [style=rededge, bend left, looseness=0.50] (21) to (24.center);
		\draw [style=rededge] (29.center) to (28);
		\draw [style=rededge, bend right, looseness=0.50] (28) to (30.center);
		\draw [style=rededge, bend left, looseness=0.50] (28) to (31.center);
	\end{pgfonlayer}
\end{tikzpicture}}
\end{center}
where $P$ and $Q$ range over processes, $x\in Act$ and $y\in A\cup\bar A$ and we identify $\bar{\bar y}\coloneqq y$. The structural isomorphisms of the monoidal category have the same form as the structural isomorphisms of $\Red$, and correspond to the largest bisimulation. The other morphisms in $\LTS$ model those rewrite rules of the usual LTS semantics that are derivable via our restricted set of deduction rules.

There is a monoidal functor $I:\Red\rightarrow\LTS$, whose action on objects is defined as $0\mapsto 0,\uparrow{\tau}$, $x.P\mapsto P,\uparrow{\tau}$ and $P\parallel Q\mapsto (I(P),I(Q))$. For morphisms, $I$ takes each structural isomorphism in $\Red$ to the corresponding isomorphism in $\LTS$, and the morphism $R$ to
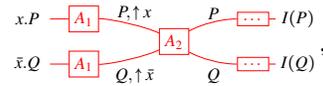
\begin{center}
\scalebox{.6}{\begin{tikzpicture}
	\begin{pgfonlayer}{nodelayer}
		\node [style=none] (0) at (-5.25, 1.25) {};
		\node [style=none] (1) at (-5.25, -0.75) {};
		\node [style=none] (2) at (-1.5, 1.5) {$P,\uparrow x$};
		\node [style=none] (3) at (-1.5, -1.25) {$Q,\uparrow\bar x$};
		\node [style=none] (4) at (2, 1.5) {$P$};
		\node [style=none] (5) at (2, -1.25) {$Q$};
		\node [style=none] (6) at (3, 1.25) {};
		\node [style=none] (7) at (3, -0.75) {};
		\node [style=redbox] (8) at (-3.75, 1.25) {$A_1$};
		\node [style=redbox] (9) at (-3.75, -0.75) {$A_1$};
		\node [style=redbox] (10) at (0.25, 0.25) {$A_2$};
		\node [style=none] (11) at (-6.25, 1.25) {$x.P$};
		\node [style=none] (12) at (-6.25, -0.75) {$\bar x.Q$};
		\node [style=none] (13) at (4.75, 1.25) {};
		\node [style=none] (14) at (4.75, -0.75) {};
		\node [style=none] (15) at (5.75, 1.25) {$I(P)$};
		\node [style=none] (16) at (5.75, -0.75) {$I(Q)$};
		\node [style=redbox] (17) at (3.75, 1.25) {$\dots$};
		\node [style=redbox] (18) at (3.75, -0.75) {$\dots$};
	\end{pgfonlayer}
	\begin{pgfonlayer}{edgelayer}
		\draw [style=rededge] (0.center) to (8);
		\draw [style=rededge] (1.center) to (9);
		\draw [style=rededge, bend left=15, looseness=0.50] (8) to (10);
		\draw [style=rededge, bend right=15, looseness=0.50] (9) to (10);
		\draw [style=rededge, bend left=15, looseness=0.75] (10) to (6.center);
		\draw [style=rededge, bend right=15, looseness=0.75] (10) to (7.center);
		\draw [style=rededge] (6.center) to (17);
		\draw [style=rededge] (17) to (13.center);
		\draw [style=rededge] (7.center) to (18);
		\draw [style=rededge] (18) to (14.center);
	\end{pgfonlayer}
\end{tikzpicture}},
\end{center}
where the dots refer to an appropriate decomposition of $P$ and $Q$ into $I(P)$ and $I(Q)$.

We use the functor $I$ to view the LTS semantics as the lower level language that explains the reduction semantics. For instance, we can explain the rewrite rule~\eqref{eqn:CCS-derivation} by just moving its derivation in $\Red$
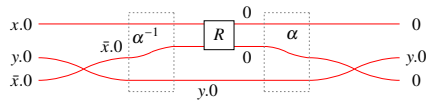
\begin{center}
\scalebox{.6}{\begin{tikzpicture}
	\begin{pgfonlayer}{nodelayer}
		\node [style=none] (20) at (-8.75, 1.5) {$x.0$};
		\node [style=none] (66) at (-8.75, 0) {$y.0$};
		\node [style=none] (67) at (-8.75, -1) {$\bar x.0$};
		\node [style=none] (68) at (-8, 0) {};
		\node [style=none] (69) at (-8, -1) {};
		\node [style=none] (70) at (-6, -0.5) {};
		\node [style=none] (71) at (-4, 0) {};
		\node [style=none] (72) at (-4, -1) {};
		\node [style=none] (73) at (-0.5, -1.5) {$y.0$};
		\node [style=none] (74) at (-4.75, 0.5) {$\bar x.0$};
		\node [style=none] (78) at (-3, 0.25) {};
		\node [style=none] (79) at (-8, 1.5) {};
		\node [style=none] (80) at (-2, 0.5) {};
		\node [style=none] (81) at (8, 1.5) {};
		\node [style=bdytextbox] (83) at (0, 1) {$\phantom{|}R\phantom{|}$};
		\node [style=none] (86) at (1.25, 2) {$0$};
		\node [style=none] (87) at (1.25, 0) {$0$};
		\node [style=none] (88) at (2, 0.5) {};
		\node [style=none] (89) at (3, 0.25) {};
		\node [style=none] (92) at (4, 0) {};
		\node [style=none] (93) at (4, -1) {};
		\node [style=none] (94) at (6, -0.5) {};
		\node [style=none] (95) at (8, 0) {};
		\node [style=none] (96) at (8, -1) {};
		\node [style=none] (97) at (8.75, 0) {$y.0$};
		\node [style=none] (98) at (8.75, 1.5) {$0$};
		\node [style=none] (99) at (8.75, -1) {$0$};
		\node [style=none] (100) at (-4, 2) {};
		\node [style=none] (101) at (-2, 2) {};
		\node [style=none] (102) at (-4, -1.5) {};
		\node [style=none] (103) at (-2, -1.5) {};
		\node [style=none] (104) at (-3.25, 1) {$\alpha^{-1}$};
		\node [style=none] (105) at (2, 2) {};
		\node [style=none] (106) at (4, 2) {};
		\node [style=none] (107) at (2, -1.5) {};
		\node [style=none] (108) at (4, -1.5) {};
		\node [style=none] (109) at (3.25, 1) {$\alpha$};
	\end{pgfonlayer}
	\begin{pgfonlayer}{edgelayer}
		\draw [style=rededge, bend left=15, looseness=0.75] (68.center) to (70.center);
		\draw [style=rededge, bend left=15, looseness=0.75] (70.center) to (69.center);
		\draw [style=rededge, bend left=15, looseness=0.75] (70.center) to (71.center);
		\draw [style=rededge, bend right=15, looseness=0.75] (70.center) to (72.center);
		\draw [style=rededge, bend right=15, looseness=0.75] (71.center) to (78.center);
		\draw [style=rededge] (79.center) to (81.center);
		\draw [style=rededge, bend left=15, looseness=0.50] (78.center) to (80.center);
		\draw [style=rededge, bend left=15, looseness=0.75] (88.center) to (89.center);
		\draw [style=rededge] (80.center) to (88.center);
		\draw [style=rededge, bend left=15, looseness=0.75] (92.center) to (94.center);
		\draw [style=rededge, bend left=15, looseness=0.75] (94.center) to (93.center);
		\draw [style=rededge, bend left=15, looseness=0.75] (94.center) to (95.center);
		\draw [style=rededge, bend right=15, looseness=0.75] (94.center) to (96.center);
		\draw [style=rededge, bend right=15, looseness=0.50] (89.center) to (92.center);
		\draw [style=rededge] (72.center) to (93.center);
		\draw [style=gray dashed] (100.center) to (101.center);
		\draw [style=gray dashed] (101.center) to (103.center);
		\draw [style=gray dashed] (103.center) to (102.center);
		\draw [style=gray dashed] (102.center) to (100.center);
		\draw [style=gray dashed] (105.center) to (106.center);
		\draw [style=gray dashed] (106.center) to (108.center);
		\draw [style=gray dashed] (108.center) to (107.center);
		\draw [style=gray dashed] (107.center) to (105.center);
	\end{pgfonlayer}
\end{tikzpicture}}
\end{center}
through the window, that is, essentially by applying $I$. In this case, we are also able to give a counterfactual explanation:
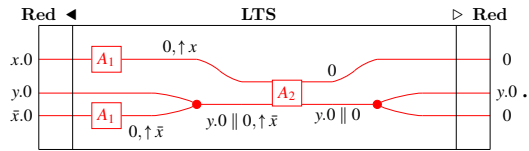
\begin{center}
\scalebox{.6}{\begin{tikzpicture}
	\begin{pgfonlayer}{nodelayer}
		\node [style=none] (0) at (-10.5, 1.5) {$x.0$};
		\node [style=none] (1) at (-10.5, 0) {$y.0$};
		\node [style=none] (2) at (-10.5, -1) {$\bar x.0$};
		\node [style=none] (3) at (-9.75, 0) {};
		\node [style=none] (6) at (10.25, 1.5) {};
		\node [style=none] (7) at (10.25, 0) {};
		\node [style=none] (8) at (10.25, -1) {};
		\node [style=none] (9) at (11, 0) {$y.0$};
		\node [style=none] (10) at (11, 1.5) {$0$};
		\node [style=none] (11) at (11, -1) {$0$};
		\node [style=redbox] (12) at (-6.75, 1.5) {$A_1$};
		\node [style=none] (13) at (-9.75, 1.5) {};
		\node [style=none] (14) at (-2.75, 1.5) {};
		\node [style=none] (16) at (-3.5, 2) {$0,\uparrow x$};
		\node [style=redbox] (17) at (-6.75, -1) {$A_1$};
		\node [style=none] (18) at (-9.75, -1) {};
		\node [style=none] (20) at (-5, -1.75) {$0,\uparrow{\bar x}$};
		\node [style=rednode] (21) at (-2.75, -0.5) {};
		\node [style=none] (23) at (-4.75, 0) {};
		\node [style=none] (24) at (-4.75, -1) {};
		\node [style=none] (25) at (-0.75, -1.25) {$y.0\parallel 0,\uparrow{\bar x}$};
		\node [style=redbox] (26) at (1.25, 0) {$A_2$};
		\node [style=none] (27) at (-0.75, 0.5) {};
		\node [style=none] (28) at (-0.75, -0.5) {};
		\node [style=none] (31) at (-1.75, 1) {};
		\node [style=none] (32) at (3.25, 1) {$0$};
		\node [style=none] (33) at (3.5, -1) {$y.0\parallel 0$};
		\node [style=rednode] (34) at (5.25, -0.5) {};
		\node [style=none] (35) at (7.25, 0) {};
		\node [style=none] (36) at (7.25, -1) {};
		\node [style=none] (37) at (3.25, -0.5) {};
		\node [style=none] (39) at (5.25, 1.5) {};
		\node [style=none] (40) at (3.25, 0.5) {};
		\node [style=none] (41) at (4.25, 1) {};
		\node [style=none] (42) at (-9.75, 3) {};
		\node [style=none] (43) at (10.25, 3) {};
		\node [style=none] (44) at (-9.75, -2.5) {};
		\node [style=none] (45) at (10.25, -2.5) {};
		\node [style=none] (46) at (-8.25, 3) {};
		\node [style=none] (47) at (-8.25, -2.5) {};
		\node [style=none] (48) at (8.75, 3) {};
		\node [style=none] (49) at (8.75, -2.5) {};
		\node [style=none] (50) at (-8.25, 3.5) {$\refine$};
		\node [style=none] (51) at (8.75, 3.5) {$\coarsen$};
		\node [style=none] (52) at (-9.75, 3.5) {$\Red$};
		\node [style=none] (53) at (10.25, 3.5) {$\Red$};
		\node [style=none] (54) at (0, 3.5) {$\LTS$};
	\end{pgfonlayer}
	\begin{pgfonlayer}{edgelayer}
		\draw [style=rededge] (13.center) to (12);
		\draw [style=rededge] (12) to (14.center);
		\draw [style=rededge] (18.center) to (17);
		\draw [style=rededge, bend right=15, looseness=0.50] (21) to (23.center);
		\draw [style=rededge, bend left=15, looseness=0.50] (21) to (24.center);
		\draw [style=rededge] (3.center) to (23.center);
		\draw [style=rededge] (17) to (24.center);
		\draw [style=rededge, bend left=15, looseness=0.75] (14.center) to (31.center);
		\draw [style=rededge, bend right=15, looseness=0.75] (31.center) to (27.center);
		\draw [style=rededge] (28.center) to (21);
		\draw [style=rededge, bend left=15, looseness=0.50] (34) to (35.center);
		\draw [style=rededge, bend right=15, looseness=0.50] (34) to (36.center);
		\draw [style=rededge] (37.center) to (34);
		\draw [style=rededge] (28.center) to (37.center);
		\draw [style=rededge, bend right=15, looseness=0.75] (39.center) to (41.center);
		\draw [style=rededge, bend left=15, looseness=0.75] (41.center) to (40.center);
		\draw [style=rededge] (27.center) to (40.center);
		\draw [style=rededge] (39.center) to (6.center);
		\draw [style=rededge] (35.center) to (7.center);
		\draw [style=rededge] (36.center) to (8.center);
		\draw [style=edge] (42.center) to (44.center);
		\draw [style=edge] (44.center) to (45.center);
		\draw [style=edge] (45.center) to (43.center);
		\draw [style=edge] (43.center) to (42.center);
		\draw [style=edge] (46.center) to (47.center);
		\draw [style=edge] (48.center) to (49.center);
	\end{pgfonlayer}
\end{tikzpicture}}.
\end{center}
The above diagram is indeed a counterfactual explanation (see the discussion in Section~\ref{sec:explanations}) of the rewrite rule~\eqref{eqn:CCS-derivation}: (1) the rewrite rule is an internal morphism in $\Red$, (2) every non-identity internal morphism in the diagram is contained in $\LTS$, which is strictly below $\Red$ in the partial order of the layered prop, (3) there are no $2$-cells between the rewrite rule and the diagram. To see that (3) is indeed the case, note that there are in fact no $2$-cells having the above diagram as either domain or codomain (one can see this by going through the generators of $2$-cells of a layered prop one by one).

The fact that there is a counterfactual explanation of the rewrite rule~\eqref{eqn:CCS-derivation} shows that it is not {\em necessary} to invoke the (analogue of) rule $R$ in its derivation at the level of LTS. This observation allows us to show neatly that LTS semantics is more flexible than the reduction semantics, in the sense that there are more derivations of the same transitions. Note that the counterfactual explanation does not need to be more complex than an ordinary explanation: in this case it is in fact more direct, in the sense that it shows that there is an actual labelled transition, while the explanation obtained by translating the diagram in $\Red$ merely shows that there is a labelled transition up to the largest bisimulation.

\section{Conclusions and Future Work}
We have taken the first steps towards developing a mathematical framework for formalising explanations. Explanations in a category theoretic context usually take the form of a functor, whose domain is thought of as syntax and codomain as semantics. Our approach differs from this: in a layered prop, there are several possible translations to different levels, which are nonetheless syntactically represented in the same language (that is, within one layered prop). A layered prop allows one to easily work with different theories describing the same phenomenon, and, importantly, allows for partial translations instead of having to translate the full diagram, as we have illustrated with the examples. We have also observed how counterfactual processes arise naturally within layered props: these are those processes that ``look like'' a translation without being one. Furthermore, the examples show that the same abstract principles hold in areas as distant as biology, electrical circuit theory and concurrency theory. Layered props can thus indeed be conceived as the initial stage of a general mathematical theory of explanations.

On the mathematical level, the next phase of developing the theory is to explore the precise connection of layered props to pointed profunctors. Currently, there is a canonical $2$-functor which translates a layered prop generated by a system of monoidal categories to the category of pointed profunctors which preserves the axioms of a layered prop. One way to proceed would be to characterise the image of this functor, thus identifying a subcategory of $\Prof_*$ to which a given layered prop is equivalent. Another mathematical aspect that is important for practical applications is to modify the definition of a layered prop to allow for non-strictly associative monoidal categories, as for instance described diagrammatically in~\cite{non-strict-monoidal}. As briefly remarked in Section~\ref{sec:electrical-circuits}, the current semantics cannot adequately handle the important special case when the translation functor is faithful. This suggests that the current interpretation of the $2$-cells as natural transformations is too restrictive, and some other notion of $2$-cells for pointed profunctors should be used. In order to connect layered props to known structures, it would also be useful to express them as a Grothendieck construction.

Even though it was beyond the scope of this paper, we believe it is important to connect our work with the philosophy of science literature on explanations. Since the initial motivation for our work comes from biology, it is particularly interesting to see how ideas on explanations and causality in biology fit our framework. For instance, one of the main motivations of Robert Rosen for introducing the theoretical framework of {\em relational biology} was to put the {\em function} of an organism on equal grounding with the {\em mechanism} that underlies it \cite{rosen-life-itself}. This can be modelled within a layered prop: reductive and functional explanations are {\em a priori} completely symmetric, and in any case equally well-defined.

Several systems with multiple layers are known in the applied category theory literature. In addition to the already discussed~\cite{non-strict-monoidal} and~\cite{electrical-circuits} (Section~\ref{sec:electrical-circuits}), we mention the formalism of {\em hierarchical petri nets}~\cite{hierarchical-petri-nets}, and Rom\' an's notion of an {\em open diagram}~\cite{roman}. All of these rely on an intuitive notion composing processes at different levels, and hence we plan to explore them using layered props.

\paragraph{Acknowledgements} The credit for the original idea for the ``calculus of refinement and coarsening'', as well as for the examples featuring chemical reactions and CCS goes to Jean Krivine. We also thank him for feedback and inspiring discussions during various stages of this work. We thank Samson Abramsky for discussing and giving feedback on ideas that lead to this work. We thank Cole Comfort for conversations and for pointing us towards the literature on internal string diagrams. LL thanks Jamie Vicary, Nick Hu, Alex Rice, Calin Tataru and Ioannis Markakis for an opportunity to present and discuss an early version of this work.

\nocite{*}
\bibliographystyle{eptcs}
\bibliography{bibliography}

\appendix
\mathversion{normal5}
\section{Profunctors and Pointed Profunctors}\label{sec:profunctors}

In order to fix notational conventions, we recall the standard definition of profunctors. We also define the not-so-standard category of pointed profunctors. We state the results about (pointed) profunctors needed in the main body of the paper, mostly without proof.

\subsection{Profunctors}
We follow Loregian~\cite{loregian} in our discussion of profunctors and coends.
\begin{definition}[Bicategory of profunctors]
Define the bicategory of {\em profunctors} $\Prof$ as follows.
\begin{itemize}[topsep=0pt,itemsep=-1ex,partopsep=1ex,parsep=1ex]
\item the $0$-cells are (small) categories,
\item the $1$-cells, denoted by $\cat C\srarrow\cat D$, are functors
$$\cat C^{op}\times\cat D\rightarrow\Set,$$
\item the $2-cells$ are natural transformations $\alpha:F\Rightarrow G$,
\item the composition
$$c_{\cat A,\cat B,\cat C} : \Prof(\cat A,\cat B)\times\Prof(\cat B,\cat C)\rightarrow\Prof(\cat A,\cat C)$$
takes profunctors $F:\cat A\srarrow\cat B$ and $G:\cat B\srarrow\cat C$ to the coend $G\circ F=\int^{B}F(-,B)\times G(B,=)$. Explicitly, we define
$$(G\circ F)(A,C)\coloneqq\int^{B\in\cat B}F(A,B)\times G(B,C).$$
\end{itemize}
\end{definition}

There is a bifunctor
$$\times : \Prof\times\Prof\rightarrow\Prof$$
defined as the product functor of $n$-cells for each $n=0,1,2$ which equips $\Prof$ with a symmetric monoidal structure.

Given a 2-category $\cat K$, let us write $\cat K^{op}$ for the 2-category whose 0-cells and 2-cells are those of $\cat K$ and whose 1-cells are the reversed 1-cells of $\cat K$, that is, for all 0-cells $A$ and $B$ we have $\cat K^{op}(A,B)=\cat K(B,A)^{op}$. Similarly, we write $\cat K^{co}$ for the 2-category whose 0-cells and 1-cells are those of $\cat K$ and whose 2-cells are the reversed 2-cells of $\cat K$, that is, for all 0-cells $A$ and $B$ we have $\cat K^{co}(A,B)=\cat K(A,B)^{op}$.

There are two ways to embed the 2-category $\Cat$ into $\Prof$: one is contravariant on the 1-cells, the other on the 2-cells. Both embeddings are identity on objects. The embedding
$$\embedup{-} : \Cat^{co}\rightarrow\Prof$$
takes a functor $F:\cat C\rightarrow\cat D$ to the profunctor $\embedup F:\cat C\srarrow\cat D$ defined on objects by $\embedup F(C,D)\coloneqq\cat D(FC,D)$, and a natural transformation $\eta:F\rightarrow G$ to the natural transformation $\embedup G\rightarrow\embedup F$ whose $(C,D)$-component is given by $-\circ\eta_C$.

Dually, the embedding
$$\embeddn - : \Cat^{op}\rightarrow\Prof$$
takes a functor $F:\cat C\rightarrow\cat D$ to the profunctor $\embeddn F:\cat D\srarrow\cat C$ defined on objects by $\embeddn F(C,D)\coloneqq\cat D(D,FC)$, and a natural transformation $\eta:F\rightarrow G$ to the natural transformation $\embedup F\rightarrow\embedup G$ whose $(C,D)$-component is given by $\eta_C\circ -$.

Both $\embedup -$ and $\embeddn -$ are 2-functors, locally fully faithful and for every functor $F$ the 1-cell $\embedup F$ is the left adjoint to $\embeddn F$ in the bicategory $\Prof$ (see section~5.1 of Loregian~\cite{loregian} for the details).

\begin{proposition}\label{prop:monoidal-2-functors}
  Both $\embedup - : \Cat^{co}\rightarrow\Prof$ and $\embeddn - : \Cat^{op}\rightarrow\Prof$ are monoidal 2-functors.
\end{proposition}
\begin{proof}
  We prove the result for $\embedup -$: the argument for $\embeddn -$ is dual.

  Since the embedding is identity on objects, the monoidal product of 0-cells (which is just the cartesian product of categories) is preserved.

  For 1-cells, let $F:\cat C\rightarrow\cat D$ and $G:\cat C'\rightarrow\cat D'$ be functors. We wish to show that $\embedup{F\times G}\simeq\embedup F\times\embedup G$. We compute as follows:
  \begin{align*}
    \embedup{F\times G}(C,C';D,D') &= \cat D\times\cat D'(F\times G(C,C'), (D,D')) \\
                                   &= \cat D\times\cat D'((FC,GC'), (D,D')) \\
                                   &= \cat D(FC,D)\times\cat D'(GC',D') \\
                                   &= \embedup F(C,D)\times\embedup G(C',D') \\
                                   &= \left(\embedup F\times\embedup G\right)(C,C';D,D'),
  \end{align*}
  whence it follows that $\embedup{F\times G}$ and $\embedup F\times\embedup G$ agree on objects. The fact that they agree on morphisms is a similar computation.

  For 2-cells, let $F,F':\cat C\rightarrow\cat D$ and $G,G':\cat C'\rightarrow\cat D'$ all be functors. Given natural transformations $\eta:F\rightarrow F'$ and $\mu:G\rightarrow G'$, we wish to show that $\embedup{\eta\times\mu}=\embedup{\eta}\times\embedup{\mu}$. This follows by observing that their components coincide:
  \begin{equation*}
    \embedup{\eta\times\mu}_{(C,C';D,D')} = -\circ (\eta\times\mu)_{(C,C')}
                                          = (-\circ\eta_C)\times (-\circ\mu_{C'})
                                          = \embedup{\eta}_{C,D}\times\embedup{\mu}_{C',D'}
                                          = \left(\embedup{\eta}\times\embedup{\mu}\right)_{C,C';D,D'}.
  \end{equation*}
\end{proof}

\subsection{Pointed Profunctors}
\begin{definition}[Pointed profunctors]
Define the bicategory of {\em pointed profunctors} $\Prof_*$ as follows:
\begin{itemize}[topsep=0pt,itemsep=-1ex,partopsep=1ex,parsep=1ex]
\item the $0$-cells are pairs $(\cat C, c)$ of a (small) category $\cat C$ and an object $c\in\Ob(\cat C)$,
\item the $1$-cells $(P,f) : (\cat C, c)\rightarrow (\cat D, d)$ consist of a profunctor $P : \cat C\srarrow\cat D$, that is, a functor
$$P : \cat C^{op}\times\cat D\rightarrow\Set,$$
together with an element $f\in P(c,d)$,
\item the $2$-cells $\alpha : (P,f) \rightarrow (Q,g)$ are natural transformations $\alpha : P\Rightarrow Q$ such that $\alpha_{c,d}(f)=g$,
\item the composition of $(P,f) : (\cat C, c)\rightarrow (\cat D, d)$ and $(Q,g) : (\cat D, d)\rightarrow (\cat E, e)$ is given by $(Q\circ P, [f,g])$, where $\circ$ is the composition of profunctors and $[f,g]$ the equivalence class of the pair $(f,g)$ in $(Q\circ P)(c,e)$.
\end{itemize}
\end{definition}

Note that a pointed hom-functor $(\cat C(-,-), f) : (\cat C, c)\rightarrow (\cat C, c')$ is precisely a morphism $f:c\rightarrow c'$. Thus we will simply write the hom-functor $(\cat C(-,-), f)$ as $f$. For a category $\cat C$, we define an assignment
\begin{align*}
z_{\cat C} : \cat C &\rightarrow \Prof_* \\
c &\mapsto (\cat C, c) \\
(f : c\rightarrow c') &\mapsto (f : (\cat C, c)\rightarrow (\cat C, c'))
\end{align*}
\begin{proposition}\label{prop:embed-omega}
The assignment $z_{\cat C}$ is a pseudofunctor (when $\cat C$ is taken to have the trivial bicategory structure).
\end{proposition}
\begin{proof}
We first show that $z_{\cat C}$ preserves composition. Thus let $f : c\rightarrow d$ and $g : d\rightarrow e$ be morphisms in $\cat C$. First, $\cat C(-,-)\circ\cat C(-,-) \simeq \cat C(-,-)$ since the hom-profunctor is the identity profunctor. Observe that such an isomorphism is given by $\int^a \cat C(c,a)\times\cat C(a,e)\xrightarrow{\sim}\cat C(c,e)$ given by $[n,m]\mapsto m\circ n$ (this is well-defined). Thus in particular $[f,g]\mapsto gf$, whence
$$z_{\cat C}(g)\circ z_{\cat C}(f) = (\cat C(-,-), g)\circ (\cat C(-,-), f) \simeq (\cat C(-,-), gf) = z_{\cat C}(gf).$$

From the above it follows that $\id_{c} : (\cat C, c)\rightarrow (\cat C, c)$ is the identity on $(\cat C, c)$, so that $z_{\cat C}$ preserves the identities.
\end{proof}

There is a pseudofunctor
$$\times : \Prof_*\times\Prof_*\rightarrow\Prof_*$$
defined as
\begin{itemize}[topsep=0pt,itemsep=-1ex,partopsep=1ex,parsep=1ex]
\item $(C,c)\times (D,d)\coloneqq (C\times D, (c,d))$ on the $0$-cells,
\item $(P,f)\times (Q,g)\coloneqq (P\times Q, (f,g))$ on the $1$-cells,
\item the product of natural transformations on the $2$-cells.
\end{itemize}
Writing $\one$ for the terminal category and $\bullet$ for its unique object, we have the following:
\begin{proposition}
$(\Prof_*,\times,(\one,\bullet))$ is a symmetric monoidal bicategory.
\end{proposition}

\section{Semantic Properties of Layered Props}\label{sec:proofs}

We discuss the properties that the interpretation functor $\mathcal I:\mathcal L(\Omega)\rightarrow\Prof_*$ has. Throughout the section, we assume that $\Omega$ is a system of monoidal categories.

We begin by proving that $\mathcal I$ is indeed a pointed profunctor model (Proposition~\ref{prop:monoidal-prof-model}).
\begin{proof}[Proof of Proposition~\ref{prop:monoidal-prof-model}]
The equalities of morphisms for each category $\omega\in\Omega$ are preserved by Proposition~\ref{prop:embed-omega}. The unit and counit maps in Figure~\ref{fig:axioms-pants2} are preserved and the triangle equalities for them hold since we have defined each pair of profunctors as an adjoint pair. Since all the internal morphisms are identities, there is nothing to show for the composition of the internal morphisms.

All the rules in Figure~\ref{fig:axioms-pants3} follow from the fact that each category and functor in $\Omega$ is monoidal and that both $\embedup -$ and $\embeddn -$ are monoidal $2$-functors (Proposition~\ref{prop:monoidal-2-functors}). For example, by (strict) associativity we have that $\otimes(\id\times\otimes)=\otimes(\otimes\times\id)$ in $\Cat$. We get the desired equations by applying the embeddings:
  \begin{align*}
    \embedup{\otimes}\circ(\embedup{\otimes}\times\id) = \embedup{\otimes}\circ(\id\times\embedup{\otimes})\quad &\texttt{assoc}, \\
    (\id\times\embeddn{\otimes})\circ\embeddn{\otimes} = (\embeddn{\otimes}\times\id)\circ\embeddn{\otimes}\quad &\texttt{coassoc}.
  \end{align*}

  It remains to show that the rules in Figure~\ref{fig:axioms-pants1} are preserved. These are the only rules with a non-trivial internal structure. Observe that all these rules are either of the form
  $$(\embedup f,\id_{f\beta})\circ\sigma\simeq f\sigma\circ (\embedup f,\id_{f\alpha})\qquad\textrm{ or }\qquad\sigma\circ (\embeddn f, \id_{f\beta})\simeq (\embeddn f, \id_{f\alpha})\circ f\sigma$$
  for some functor $f:\omega\rightarrow\tau$ and some morphism $\sigma:\alpha\rightarrow\beta$. We show the isomorphism on the left. First, at the level of profunctors the isomorphism holds since hom-functors are the identities. It remains to show that $[\sigma,\id_{f\beta}]\sim[\id_{f\alpha},f\sigma]$ under this isomorphism. To this end, note that the left-hand side evaluates via the isomorphism
  $$\int^{\gamma\in\omega}\omega(\alpha,\gamma)\times\omega(f\gamma,f\beta)\simeq\omega(f\alpha,f\beta)\qquad [h,k]\mapsto k\circ fh$$
  to $f\sigma$. Similarly, the right-hand side evaluates via the isomorphism
  $$\int^{\gamma\in\tau}\omega(f\alpha,\gamma)\times\omega(\gamma,f\beta)\simeq\omega(f\alpha,f\beta)\qquad [h,k]\mapsto k\circ h$$
  also to $f\sigma$, whence the desired identification follows. The argument for the rules of the second form is dual.
\end{proof}

The following proposition shows that we can detect properties of monoidal categories in a layered prop. This observation is not relevant for the examples that we discuss in this work, yet it is important for the development of the general theory of layered props.
\begin{proposition}
If both $\tau,\omega\in\Omega$ are monoidal closed ({\em resp. coclosed}) and $f:\omega\rightarrow\tau$ in $\Omega$ is also monoidal closed ({\em resp. coclosed}), then the interpretation $\mathcal I$ preserves the $2$-cell {\normalfont\texttt{C}} ({\em resp.} {\normalfont\texttt{coC}}) in Figure~\ref{fig:pants-extra-conditions}.
\end{proposition}
\begin{proof}
  For \texttt{C}, we have to show that
  $$\embedup{\otimes}\circ(\embeddn f\times\id)\simeq\embeddn f\circ\embedup{\otimes}\circ(\id\times\embedup f).$$
  Both profunctors are of the type $\tau\times\omega\srarrow\omega$. Let us compute both sides on the triple of objects $(D,C,C')$. The right-hand side computes to
  \begin{align*}
    \int^{B,E\in\tau}\tau(fC,B)\times\tau(D\otimes B,E)\times\tau(E,FC') \simeq \tau(D\otimes FC,FC'),
  \end{align*}
  while in order to reduce the left-hand side we use the monoidal closed structure:
  \begin{align*}
    \int^{A\in\omega}\tau(D,FA)\times\omega(A\otimes C, C') &\simeq \int^{A\in\omega}\tau(D,FA)\times\omega(A, [C,C']) \\
                                                              &\simeq\tau(D, F[C,C']) \\
                                                              &\simeq\tau(D, [FC,FC']) \\
                                                              &\simeq\tau(D\otimes FC, FC').
  \end{align*}
  Since these agree and all isomorphisms are natural, we have the desired isomorphism. The argument for \texttt{coC} is dual, using that the categories and functors are monoidal coclosed.
\end{proof}

\begin{figure}[h]
    \centering
    \resizebox{\textwidth}{!}{
        \input{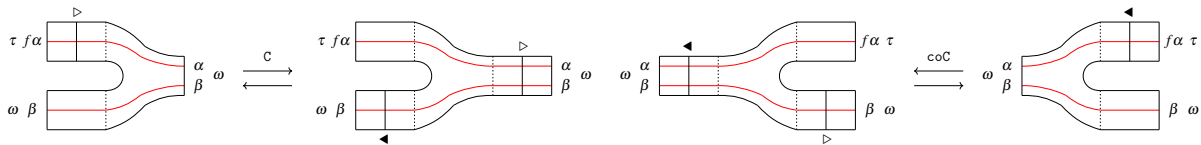}
    }
    \caption{Monoidal (co)closure equations.\label{fig:pants-extra-conditions}}
\end{figure}

\end{document}